\newif\ifllncsloaded
\newif\ifieetranloaded
\newif\ifacmartloaded
\newif\ifarxiveloaded

\llncsloadedfalse
\ieetranloadedfalse
\acmartloadedtrue
\arxiveloadedfalse

\ifacmartloaded
\documentclass[acmsmall]{acmart}
\else
\ifllncsloaded
\documentclass[runningheads]{llncs}
\else
\ifieetranloaded
\documentclass[lettersize,journal]{IEEEtran}
\else
\documentclass{amsart}       
\fi
\fi
\fi

\usepackage[utf8]{inputenc}
\usepackage{xparse}
\usepackage[T1]{fontenc}
\usepackage{microtype}
\usepackage{xspace}
\usepackage{relsize}
\usepackage[normalem]{ulem}
\usepackage{ragged2e}           
\usepackage{url}

\usepackage[scaled=0.88]{helvet}

\usepackage[font=small,labelfont=bf]{caption}
\usepackage[caption=false, labelformat=simple]{subfig}

\usepackage{array}
\newcolumntype{H}{>{\setbox0=\hbox\bgroup}c<{\egroup}@{}}
\usepackage{booktabs}
\usepackage{tabularx}
\usepackage[flushleft]{threeparttable}
\usepackage{multirow}

\newcolumntype{R}{>{\raggedleft\arraybackslash}X}
\newcolumntype{D}{>{$}c<{$}}
\newcolumntype{Q}{>{$}l<{$}}
\usepackage[group-separator={,}]{siunitx}
\usepackage{xifthen}            
\usepackage{mathtools}
\ifllncsloaded\else
\usepackage{amsmath}
\usepackage{amsthm}
\fi
\ifieetranloaded
\usepackage{amssymb}
\fi
\usepackage{halloweenmath}
\usepackage[MnSymbols,rawdefs]{sty/u8chars}
\usepackage[%
enums,%
scfg-noassertions,%
fsm-noassertions,%
]{sty/symbolic-notations}
\usepackage{galois}

\usepackage[ruled,vlined]{algorithm2e}
\usepackage{listings}
\usepackage{wrapfig}
\usepackage[inline,shortlabels]{enumitem}
\usepackage{graphicx}
\usepackage{tikz}

\usepackage[textsize=footnotesize,textwidth=1.5cm]{todonotes}
\usepackage{comment}            
\usepackage{xargs}

\newcommand{\itodo}[2][]{\todo[inline,caption={2do}, #1]{%
    \begin{minipage}{\textwidth-4pt}#2\end{minipage}}}
\ifarxiveloaded
\usepackage[square,comma,numbers,sort,sort&compress]{natbib}

\fi

\ifacmartloaded\else
\usepackage[
  colorlinks,
  linkcolor={red!50!black},
  citecolor={blue!50!black},
  urlcolor={blue!80!black},
]{hyperref}
\fi

\usepackage{fancyhdr}
\usepackage[24hr]{datetime}
\AtEndDocument{\label{page:last}}
\ifacmartloaded
\fi

\ifllncsloaded
\else
\fi

\ifacmartloaded
\theoremstyle{acmplain}
\fi

\ifllncsloaded\else
\newtheorem{property}{Property}
\theoremstyle{remark}

\fi

\newcommand{\aka}		{{\itshape	    aka\/}\xspace}

\newcommand{\cf}		{{\itshape	      cf.}\xspace}

\newcommand{\eg}		{{e.g.,}\xspace}

\newcommand{\ie}		{{i.e.,}\xspace}

\newcommand{\wrt}		{{w.r.t.}\xspace}

\newcommand\B{\ensuremath{\mathcal B}\xspace}%
\providecommand\G{\ensuremath{\mathcal G}\xspace}%
\newcommand\Q{\ensuremath{\mathcal Q}\xspace}%
\renewcommand\S{\ensuremath{\mathcal S}\xspace}%
\newcommand\X{\ensuremath{\mathcal X}\xspace}%
\newcommand\Z{\ensuremath{\mathcal Z}\xspace}%
\newcommand\BB{\ensuremath{\mathbb B}\xspace}%
\newcommand\LL{\ensuremath{\mathbb L}\xspace}%
\newcommand{\ReaX}{{\sffamily ReaX}\xspace}
\newcommand{\Java}{{\sffamily Java}\xspace}

\newcommand{\soot}{{\sffamily soot}\xspace}
\newcommand{\KeY}{{\scshape   KeY}\xspace}

\newcommand{\IFSPEC}{\textsc{IFSpec}\xspace}
\newcommand\IFSpec\IFSPEC

\newcommandx\adjustedbg[3][1=0pt,2=white]{%
  \setlength{\fboxsep}{#1}%
  \colorbox{#2}{#3}%
}

\newcommand\cv[1]{\ensuremath{\mathtt{#1}}\xspace}%

\definecolor{bleudefrance}{rgb}{0.19, 0.55, 0.91}
\definecolor{burgundy}{rgb}{0.5, 0.0, 0.13}
\definecolor{darkcandyapplered}{rgb}{0.64, 0.0, 0.0}
\colorlet{hdval}{bleudefrance!80!black}
\colorlet{hvar}{darkcandyapplered}
\newcommand\hdval[1]{\ensuremath{\textcolor{hdval}{#1}}\xspace}
\newcommand\hd{\hdval{\mathsf{hd}}}

\newcommand\hdeep{\hdval{\mathsf{deep}}}
\newcommand\hshal{\hdval{\mathsf{shal}}}

\newcommand\hsyma{\hdval{\mathsf{syma}}}

\newcommand{\stms}{\ensuremath{\mathbb{S}}\xspace}
\newcommand\ipstms{\ensuremath{\mathbb{S'}}\xspace}%

\newcommand\lbl[1]{\ensuremath{\mathit{#1}}\xspace}%
\newcommand\lambdaone[2][]{\ensuremath{\ifthenelse{\isempty{#2}}{}{#1（#2）}}}
\newcommand\lambdatwo[3][]{\ensuremath{\ifthenelse{\isempty{#1#2}}{}{#1（#2, #3）}}}
\newcommand\llambdaone[1]{\ensuremath{\ifthenelse{\isempty{#1}}{}{(#1)}}}
\newcommand\llambdatwo[2]{\ensuremath{\ifthenelse{\isempty{#1#2}}{}{(#1, #2)}}}
\newcommand\subst[3]{\ensuremath{#1\left[#2 ⟼ #3\right]}}%

\newlength\myrulespace%
\newcommand\rname[1]{\textsc{#1}\xspace}%

\newif\ifflushrulesleft
\flushrulesleftfalse

\newcommand{\tuple}[1]{\ensuremath{\left(#1\right)}\xspace}%
\newcommand{\semloc}[1]{\tuple{#1}}%
\newcommand{\myfreerule}[3]{\ensuremath{\rname{#1}\dfrac{#2}{#3}}}
\newenvironment{freeruleset}{
}{}
\ifflushrulesleft
  \newenvironment{ruleset}{
    \align}{\endalign}
  \newcommand{\myrule}[3]{\rname{#1}&\dfrac{#2}{#3\hfill}\nonumber}
\else
  \newenvironment{ruleset}{
    \centering}{\\}
  \newcommand{\myrule}[3]{\myfreerule{#1}{#2}{#3}}
\fi

\newcommand{\trans}[1]{%
  \raisebox{-.5ex}{\ensuremath{\xrightarrow{#1}}}
}

\newcommand{\myoverline}[2][3]{%
  {}\mkern#1mu\overline{\mkern-#1mu#2}}
\newcommand{\myoverrightarrow}[2][1]{
  \mkern#1mu%
  \overscriptrightarrow{
    \mkern-#1mu#2
  }}

\newcommand\PrimFields[1]{\ensuremath{\mathsf{PFields}}\xspace}%
\newcommand\this{\ensuremath{\mathit{this}}\xspace}%
\newcommand\MethBody[1]{\ensuremath{\stms_{#1}}\xspace}%
\newcommand\MethRefs[1]{\ensuremath{\mathit{Refs}_{#1}}\xspace}%
\newcommand\MethDeps[1]{\ensuremath{\mathsf{Deps}\lambdaone{#1}}\xspace}%
\newcommand\TargetName[1]{\ensuremath{\mathsf{target}_{#1}}\xspace}%
\newcommand\Target[2]{\ensuremath{\TargetName{#1}(#2)}\xspace}%
\newcommand\MethFormalArgs[1]{\ensuremath{\mathit{Args}_{#1}}\xspace}%
\newcommand\elts{\red{\ensuremath{y}}\xspace}%
\newcommand\meth{\ensuremath{m'}\xspace}%
 \newcommand\s[1]{\ensuremath{⟪#1⟫}\xspace}%

\newcommand\stm{\ensuremath{\mathit{a}}\xspace}%
\newcommand\cgoto{\ensuremath{\mathtt{goto}}\xspace}%
\newcommand\cif{\ensuremath{\mathtt{if}}\xspace}%
\newcommand\creturn{\ensuremath{\mathtt{return}}\xspace}%
\newcommand\coutput[1]{\ensuremath{\mathtt{output}_{#1}}\xspace}%
\newcommand\cnew{\New}%
\newcommand\Null{\ensuremath{\mathtt{null}}\xspace}%
\newcommand\New{\ensuremath{\mathtt{new}}\xspace}%
\newcommand\sep{~|~}
\newcommand\lits{\red{\ensuremath{w}}\xspace}%

\newcommandx\SymbolicAbstractHeapDom[2][1=\hd]{\ensuremath{\mathsf{HeapDom}_{#1}}\xspace}
\newcommandx\Rels{\ensuremath{\mathcal{X}}\xspace}%
\newcommandx\HeapDomRelations[1][1=\hd]{\ensuremath{\Rels_{\mathsf{Sen}}}\xspace}%
\newcommandx\HeapDomConstRelations[1][1=\hd]{\ensuremath{\Rels_{\mathsf{Insen}}}\xspace}%
\newcommandx\ConcreteHeapDom[1]{\ensuremath{\SymbolicAbstractHeapDom[\concrete]{#1}}\xspace}

\newcommand\HeapDom[1]{\ensuremath{\mathbb{H}_{#1}}\xspace}%
\newcommand\Transformers[1]{\ensuremath{\mathbb{T}_{#1}}\xspace}%
\newcommandx\TypeAnalysis[3][1=\hd]{\ensuremath{\mathsf{CanRelate}\lambdaone{#3}}\xspace}%

\newcommandx\DeepHeapAbstractDom[1]{\SymbolicAbstractHeapDom[\hdeep]{#1}}
\newcommandx\DeepHeapDom[1]{\HeapDom[\hdeep]{#1}}
\newcommandx\DeepTransformers[1]{\Transformers[\hdeep]{#1}}

\newcommand\HeapCappedJoinProj[3]{\ensuremath{\mathsf{Upgrade}(#1,#2,#3)}\xspace}%

\newcommand\hv[1]{\ensuremath{\textcolor{hvar}{#1}}\xspace}
\newcommand\hvar{\ensuremath{\hv{\mathfrak h}}\xspace}%
\newcommand\fvar{\ensuremath{\hv{\mathfrak f}}\xspace}%
\newcommand\pval{\ensuremath{\mathtt{Prm}}\xspace}%

\newcommand\hlvlH[2][\mathclap{\phantom{\hvar}}]{\ensuremath{\ifthenelse{\isempty{#1}}{%
      \vec{#2}%
    }{%
      \myoverrightarrow{#2}\mkern-2mu%
      ^{\mathchoice%
	{\raisebox{-2pt}{$\displaystyle#1$}}%
	{\raisebox{-1pt}{$#1$}}%
	{\raisebox{-0.5pt}{$\scriptstyle#1$}}%
	{\raisebox{-0.2pt}{$\scriptscriptstyle#1$}}}%
    }}\xspace}

\newcommand\hlvl[2][\mathclap{\phantom{\hvar}}]{\ensuremath{\ifthenelse{\isempty{#1}}{%
      \vec{#2}%
    }{%
      \myoverrightarrow{#2}\mkern-2mu%
    }}\xspace}
\newcommand\flvl[3][]{\hlvl[#1]{#2}}
\newcommand\chlvl[2][\mathclap{\phantom{\hvar}}]{\hlvl[#1]{\cv{#2}}}%
\newcommand\GenericHLvlVars[1][\sHeap]{\ensuremath{\Varset
_{\hlvl L}}\xspace}%
\newcommand\GenericLvlVars[1][\sHeap]{\ensuremath{\Varset_{\plvl L}}\xspace}%
\newcommand\GenericRelSymb{\ensuremath{\smallfrown}\xspace}%
\newcommand\GenericRel[2]{\ensuremath{#1 \GenericRelSymb #2}\xspace}%
\newcommand\PrimVars{\ensuremath{P}\xspace}%

\newcommand\Varset{\ensuremath{\mathbb V}\xspace}%
\newcommand\Consset{\ensuremath{\mathbb C}\xspace}%

\newcommand\GenericRelTauto[1][\sHeap]{\ensuremath{\Consset
_{\tt}}\xspace}%
\newcommand\GenericRelUnsat[1][\sHeap]{\ensuremath{\Consset
_{\ff}}\xspace}%

\newcommandx\GenericUpdateRels[3][1=\hd]{\ensuremath{\mathsf{UpdHpRel}(#2)}\xspace}%
\newcommandx\NoHeapEffect[3][1=\hd]{\ensuremath{⦇#2⦈\xspace}}%
\newcommandx\OneHeapEffect[4][1=\hd]{\ensuremath{⦇#2, \uparrow#3⦈}\xspace}%
\newcommandx\HeapUpgrade[4][1=\hd]{\ensuremath{⦇#2 ⬿\,#3⦈_{#4}^{\hdval{#1}}}\xspace}%
\newcommandx\GenericHeapInit[3][1=\hd,2=\sHeap]{\NoHeapEffect[#1]{\mathsf{null}~#3}{#2}}%
\newcommandx\HeapInit[2][1=\sHeap]{\GenericHeapInit[\hd][#1]{#2}}
\newcommandx\DeepInit[2][1=\sHeap]{\GenericHeapInit[\hdeep][#1]{#2}}
\newcommandx\ShalInit[2][1=\sHeap]{\GenericHeapInit[\hShal][#1]{#2}}

\newcommandx\NullRef[3][1=\hd]{\NoHeapEffect[#1]{#2 = \Null}{#3}}%
\newcommandx\CopyRef[4][1=\hd]{\NoHeapEffect[#1]{#2 = #3}{#4}}%
\newcommand\LoadRef[4]{\NoHeapEffect{#1 = #2.#3}{#4}}%
\newcommand\CopyRefOp[2]{\ensuremath{#1 = #2}}%
\newcommand\LoadRefOp[3]{\ensuremath{#1 = #2.#3}} %
\newcommand\NewRefOp[2]{\ensuremath{#1 = \New}\xspace}%
\newcommand\NewRef[4]{\OneHeapEffect{\NewRefOp{#1}{#2}}{#3}{#4}}%
\newcommand\StorePrimOp[3]{\ensuremath{#1.#2 ⬿}\xspace}%
\newcommand\StoreRefOp[3]{\ensuremath{#1.#2 = #3}\xspace}%
\newcommand\StorePrim[5]{\OneHeapEffect{\StorePrimOp{#1}{#2}{#3}}{#4}{#5}}%
\newcommand\StoreRef[5]{\OneHeapEffect{\StoreRefOp{#1}{#2}{#3}}{#4}{#5}}%

\makeatletter
\newcommand{\oset}[2]{%
	{\mathop{#2}\limits^{\vbox to -0.5 \ex@{\kern-\tw@\ex@
				\hbox{\scriptsize #1}\vss}}}}
\makeatother

\newcommand\AliasRelName{\ensuremath{∼}\xspace}%
\newcommand\TransFieldAliasRelSymbol{\ensuremath{\stackrel{.*}{\hookrightarrow}}\xspace}%

\newcommand\AliasRel[3][]{\ensuremath{{\oset{\AliasRelName}{p}}_{#2,#3}}\xspace}%
\newcommand\cAliasRelProp[3][]{\AliasRel[#1]{\cv{#2}}{\cv{#3}}}
\newcommand\TransFieldAliasRel[3][]{\ensuremath{{\overset{{\TransFieldAliasRelSymbol}}{p}}_{#2,#3}}\xspace}%

\newcommand\ConnectRelSymbol{\ensuremath{\stackrel{*}{\backsim}}\xspace}%
\newcommand\ConnectRel[3][]{\ensuremath{{#2{\ConnectRelSymbol}^{#1}\ifthenelse{\isempty{#1}}{}{\!\!}#3}}\xspace}%

\newcommand\jsone{\ensuremath{\mathsf{njb}}\xspace}%
\newcommand\newloc[1]{\semloc{#1, \jsone}}%
\newcommand\guard{\ensuremath{\mathit{guard}}\xspace}%
\newcommand\plvl[1]{\ensuremath{\mathop{\myoverline[1]{#1}}}\xspace}%

\newcommand\cplvl[1]{\plvl{\cv{#1}}}%
\newcommand\pc{\ensuremath{\mathit{pc}}\xspace}%
\newcommand\Chkpt[1]{\guard}%
\newcommand\uVar{Υ\xspace}%
\newcommand\assignv{\assign_{\uVar}}%
\newcommand\nomlvl[1]{\ensuremath{\lfloor#1\rfloor_{\uVar}}\xspace}%
\newcommand\TBrch[1]{\ensuremath{\mathrm{Brch}\lambdaone{#1}}\xspace}%
\newcommand\Endua[1]{\ensuremath{\mathrm{End\mbox-ua}
  }\xspace}%
\newcommand\sVar{\ensuremath{\mathit{ret\mbox-var}}\xspace}%
\newcommand\sSpec[1]{\ensuremath{\mathsf{Ret}_{#1}}\xspace}%
\newcommand\retVar[1]{\ensuremath{#1^{\mathit{r}}}\xspace}%
\newcommand\regionName[1]{\ensuremath{\mathsf{CDR}_{#1}}\xspace}%
\newcommand\region[2]{\ensuremath{\regionName{#1}(#2)}\xspace}%
\newcommand\GotoRule{\rname{Goto}}%
\newcommand\OutputRule{\rname{Sink}}
\newcommand\AssignRule{\rname{Assign}}%
\newcommand\BranchRule{\rname{Branch}}%
\newcommand\ReturnRule{\rname{Return}}%
\newcommand\TargetSummaries[2][]{\ensuremath{\mathit{Sums}_{#2}}\xspace}%

\newcommand{\mode}{\uVar}

\newcommand{\low}{\ensuremath{\bot}\xspace }
\newcommand{\high}{\ensuremath{\top}\xspace }

\newcommand{\qstate}{q}

\newcommand\evalStm[1]{\ensuremath{
   \mathsf{Apply}_{#1}
  }\xspace}
\newcommand\NewFrame[1][]{\ensuremath{\mathrm{NewFrame}\lambdaone{#1}}\xspace}
\newcommand{\mStates}{\ensuremath{\Gamma}}

\newcommand{\mState}{\ensuremath{\gamma}\xspace}

\newcommand\concrete{\ensuremath{\mathsf{crt}}\xspace}

\newcommand{\sStates}{\ensuremath{\overline{\mStates}\xspace}}
\newcommand{\sState}{\ensuremath{\overline{\mState}}\xspace}
\newcommand{\sHeap}{\hvar}
\newcommand{\isState}[1]{\ensuremath{\sStates_{0}}\xspace}

\newcommand{\smTrans}{\ensuremath{\to}}
\newcommand{\smStates}{\ensuremath{\mathcal{Q}}\xspace}

\newcommand{\mCallRule}{\textsc{m-Call}\xspace}
\newcommand{\mStmRule}{\textsc{m-Stm}\xspace}

\newcommand{\mExitRule}{\textsc{m-Exit}\xspace}

\newcommand\evalStms[1]{\ensuremath{
		\mathsf{Apply^*} {#1}
	}\xspace}

\newcommand{\mStack}{\ensuremath{\eta}\xspace}

\newcommand\methFrame[1][m]{\ensuremath{\mathsf{f}_{#1}}\xspace}%
\newcommand\mframe[4][m]{\ensuremath{\tuple{#2, #3,#4}_{#1}}\xspace}%
\newcommand\mframeX[5]{\ensuremath{〈#2,#3,#1,#4,#5〉}\xspace}%

\ifllncsloaded\else
\DeclareRobustCommand{\qed}{%
  \ifmmode           
  \else \leavevmode\unskip\penalty9999 \hbox{}\nobreak\hfill
  \fi
  \quad\hbox{\qedsymbol}}
\fi
\makeatletter
\newenvironment{proofSketch}[1][Proof Sketch:]{\par
  \normalfont
  \topsep6\p@\@plus6\p@ \trivlist
\item[\hskip\labelsep\itshape
  #1.]\ignorespaces
}{%
  \qed\endtrivlist
}
\makeatother

\newcommand{\ttrans}[2]{\raisebox{-.7ex}{\ensuremath{\xRightarrow{\mbox{\smaller\keymathbox{#1}\,\keymathbox{#2}}}}}\xspace}%

\newcommand\lowbisim{\ensuremath{\sim_{\textit{low}}}\xspace}

\newcommand\GeneralH  {\ensuremath{\mathsf{h}}\xspace}%

\newcommand{\invars}{\ensuremath{φ}\xspace}
\newcommand{\Methods}{\ensuremath{M}\xspace}

\newcommand{\MethSums}{\ensuremath{{Sums}}\xspace}

\newcommand\HeapTypes{\ensuremath{\Gamma}\xspace}%

\newcommandx\GenericUpgradeObjLevel[4][1=\hd,2=\sHeap]{\ensuremath{\mathsf{UpdHpLev}\llambdatwo{#3}{#4}}\xspace}%
\newcommandx\HeapUpgradeObjLevel[3][1=\sHeap]{\GenericUpgradeObjLevel[\hd][#1]{#2}{#3}}%
\newcommandx\DeepUpgradeObjLevel[3][1=\sHeap]{\GenericUpgradeObjLevel[\hdeep][#1]{#2}{#3}}%
\newcommandx\ShalUpgradeObjLevel[3][1=\sHeap]{\GenericUpgradeObjLevel[\hshal][#1]{#2}{#3}}%
\newcommandx\SymaUpgradeObjLevel[3][1=\sHeap]{\GenericUpgradeObjLevel[\hsyma][#1]{#2}{#3}}%

\newcommandx\GenericBulkUpgradeFrom[3][1=\hd,2=\sHeap]{\ensuremath{\mathsf{BulkUpgr}_{#2←#3}}\xspace}%
\newcommandx\HeapBulkUpgradeFrom[2][1=\sHeap]{\GenericBulkUpgradeFrom[\hd][#1]{#2}}
\newcommandx\DeepBulkUpgradeFrom[2][1=\sHeap]{\GenericBulkUpgradeFrom[\hdeep][#1]{#2}}
\newcommandx\ShalBulkUpgradeFrom[2][1=\sHeap]{\GenericBulkUpgradeFrom[\hshal][#1]{#2}}

\newcommandx\GenericRestoreObjLevels[3][1=\hd,2=\sHeap]{\ensuremath{\mathsf{RstrLev}_{#2←#3}}\xspace}%
\newcommandx\HeapRestoreObjLevels[2][1=\sHeap]{\GenericRestoreObjLevels[\hd][#1]{#2}}
\newcommandx\DeepRestoreObjLevels[2][1=\sHeap]{\GenericRestoreObjLevels[\hdeep][#1]{#2}}
\newcommandx\ShalRestoreObjLevels[2][1=\sHeap]{\GenericRestoreObjLevels[\hshal][#1]{#2}}
\newcommandx\SymaRestoreObjLevels[2][1=\sHeap]{\GenericRestoreObjLevels[\hsyma][#1]{#2}}

\newcommandx\GenericCopyAliases[3][1=\hd,2=\sHeap]{\ensuremath{\mathsf{CopyRels}_{#2←#3}}\xspace}%
\newcommandx\HeapCopyAliases[2][1=\sHeap]{\GenericCopyAliases[\hd][#1]{#2}}
\newcommandx\DeepCopyAliases[2][1=\sHeap]{\GenericCopyAliases[\hdeep][#1]{#2}}
\newcommandx\ShalCopyAliases[2][1=\sHeap]{\GenericCopyAliases[\hshal][#1]{#2}}

\newcommandx\GenericNullRefs[3][1=,2=\sHeap]{\ensuremath{\mathsf{NullRefs}_{#2}\lambdaone{#3}}\xspace}%
\newcommandx\HeapNullRefs[2][1=\sHeap]{\GenericNullRefs[\hd][#1]{#2}}
\newcommandx\DeepNullRefs[2][1=\sHeap]{\GenericNullRefs[\hdeep][#1]{#2}}
\newcommandx\ShalNullRefs[2][1=\sHeap]{\GenericNullRefs[\hshal][#1]{#2}}

\newcommandx\GenericHeapResetRef[3][1=\hd,2=\sHeap]{\ensuremath{\GenericUpdateRels[#1]{#3 = \_}{#2}}\xspace}%
\newcommandx\HeapResetRef[2][1=\sHeap]{\GenericHeapResetRef[\hd][#1]{#2}}
\newcommandx\DeepResetRef[2][1=\sHeap]{\GenericHeapResetRef[\hdeep][#1]{#2}}
\newcommandx\ShalResetRef[2][1=\sHeap]{\GenericHeapResetRef[\hshal][#1]{#2}}
\newcommandx\SymaResetRef[2][1=\sHeap]{\GenericHeapResetRef[\hsyma][#1]{#2}}

\newcommandx\GenericHeapCopyRef[4][1=\hd,2=\sHeap]{\ensuremath{\GenericUpdateRels[#1]{\CopyRefOp{#3}{#4}}{#2}}\xspace}%
\newcommandx\HeapCopyRef[3][1=\sHeap]{\GenericHeapCopyRef[\hd][#1]{#2}{#3}}
\newcommandx\DeepCopyRef[3][1=\sHeap]{\GenericHeapCopyRef[\hdeep][#1]{#2}{#3}}
\newcommandx\ShalCopyRef[3][1=\sHeap]{\GenericHeapCopyRef[\hshal][#1]{#2}{#3}}
\newcommandx\SymaCopyRef[3][1=\sHeap]{\GenericHeapCopyRef[\hsyma][#1]{#2}{#3}}

\newcommandx\GenericHeapLoadRef[5][1=\hd,2=\sHeap]{\ensuremath{\GenericUpdateRels[#1]{\LoadRefOp{#3}{#4}{#5}}{#2}}\xspace}%
\newcommandx\HeapLoadRef[4][1=\sHeap]{\GenericHeapLoadRef[\hd][#1]{#2}{#3}{#4}}
\newcommandx\DeepLoadRef[4][1=\sHeap]{\GenericHeapLoadRef[\hdeep][#1]{#2}{#3}{#4}}
\newcommandx\ShalLoadRef[4][1=\sHeap]{\GenericHeapLoadRef[\hshal][#1]{#2}{#3}{#4}}
\newcommandx\SymaLoadRef[4][1=\sHeap]{\GenericHeapLoadRef[\hsyma][#1]{#2}{#3}{#4}}

\newcommandx\GenericHeapStoreRef[5][1=\hd,2=\sHeap]{\ensuremath{\GenericUpdateRels[#1]{\StoreRefOp{#3}{#4}{#5}}{#2}}\xspace}%
\newcommandx\HeapStoreRef[4][1=\sHeap]{\GenericHeapStoreRef[\hd][#1]{#2}{#3}{#4}}
\newcommandx\DeepStoreRef[4][1=\sHeap]{\GenericHeapStoreRef[\hdeep][#1]{#2}{#3}{#4}}
\newcommandx\ShalStoreRef[4][1=\sHeap]{\GenericHeapStoreRef[\hshal][#1]{#2}{#3}{#4}}
\newcommandx\SymaStoreRef[4][1=\sHeap]{\GenericHeapStoreRef[\hsyma][#1]{#2}{#3}{#4}}

\newcommandx\Connect[2]{\ensuremath{\mathsf{Connect}\lambdatwo{#1}{#2}}\xspace}%

\newcommandx\Share[2]{\ensuremath{\mathsf{Share}\lambdatwo{#1}{#2}}\xspace}%

\newcommand{\InRelation}{\ensuremath{\beta}\xspace}

\newcommandx\ConcCopyRef[3][1=\sHeap]{\red{\HeapCopyRef[\hd][#1]{#2}{#3}}}
\newcommandx\ConcLoadRef[4][1=\sHeap]{\red{\HeapLoadRef[#1]{#2}{#3}{#4}}}
\newcommandx\ConcStoreRef[4][1=\sHeap]{\red{\HeapStoreRef[\hd][#1]{#2}{#3}{#4}}}

\newcommand{\mHeapVal}{\ensuremath{\mathbf{h}}\xspace}
\newcommand{\VarsTypes}{\ensuremath{\Omega}\xspace}

\newcommand\SinkRule{\rname{Exit}}%

\def\cthis{\cv{this}}

\newcommand\lambdathree[4][]{\ensuremath{\ifthenelse{\isempty{#1#2#3}}{}{#1（#2, #3, #4）}}}

\setlist{noitemsep}
\setlist[1]{labelindent=\parindent} 
\newlist{compactitem}{itemize}{4}
\setlist[compactitem,1]{nolistsep,label=\textbullet}
\setlist[description]{font=\mdseries\itshape}
\newlist{mathdesc}{description}{4}
\newlist{mathdesc*}{description*}{4}
\newlist{mathpardesc}{description}{4}
\newlist{mathpardesc*}{description}{4}
\newlist{mathpardesc**}{description*}{4}
\newcommand*{\keymathbox}[1]{%
  \mdseries%
  \upshape%
  \setlength{\fboxsep}{.4pt}%
  \fcolorbox{gray}{white}{%
    \strut\ensuremath{#1}%
  }%
}%
\setlist[mathdesc]{format=\ensuremath}
\setlist[mathdesc*]{format=\ensuremath,mode=unboxed}
\setlist[mathpardesc]{format=\keymathbox,
  leftmargin=\parindent,
  labelindent=0pt,
}
\setlist[mathpardesc*]{format=\keymathbox,
  nosep,
  leftmargin=0pt,
  labelindent=\parindent,
}
\setlist[mathpardesc**]{format=\keymathbox,mode=unboxed}

\newlist{bolddescr}{description}{4}
\newcommand*{\mybolddescritem}[1]{\bfseries\upshape{#1}:}%
\setlist[bolddescr]{style=sameline, nosep, format=\mybolddescritem,
  leftmargin=0pt, labelindent=0pt,
}

\usepackage{etoolbox}

\SetCommentSty{AlgoCommandFont}
\SetFillComment
\SetKwComment{Comment}{\{\,}{\,\}}
\DontPrintSemicolon

\LinesNumbered
\makeatletter
\patchcmd{\@algocf@start}
{-1.5em}
{-.2em}
{}{}
\makeatother

\makeatletter
\newcommand{\removelatexerror}{%
  \ifieetranloaded%
  \let\@latex@error\@gobble%
  \fi%
}
\makeatother

\usetikzlibrary{fit}
\usetikzlibrary{arrows, arrows.meta}
\usetikzlibrary{decorations,decorations.pathmorphing}
\usetikzlibrary{calc}

\usetikzlibrary{external}
\AtBeginDocument{\tikzexternaldisable} 
\newif\ifusetikzexternal
\usetikzexternalfalse

\edef\defaultpgflinewidth{\the\pgflinewidth}
\tikzset{
  every path/.style = {
    line width=1pt,
    cap=round,
    join=round,
  },
  rho-highlight/.style = {
    rectangle,
    rounded corners=2pt,
    draw,
    thin,
    inner sep=0pt,
    outer sep=0pt,
  },
  rho-annot-arrow/.style = {
    thin,
    -stealth,
  },
}

\newcommand{\tikzmark}[2][base]{%
  \ensuremath{%
    \mbox{%
      \tikz[overlay,remember picture,baseline=(#2.#1)]%
      \node (#2) {};%
    }%
  }%
}

\newcommand{\lstkwstyle}{\color{blue}\bfseries}
\newcommand{\lstbasicstyle}{\fontfamily{lmvtt}\selectfont%
\upshape%
}
\newcommand{\lstbasicsize}{
\footnotesize\linespread{0.96}%
}
\newcommand{\lstinlinesize}{
}
\lstdefinestyle{numbered}{%
  numbers=left,
  numberstyle=\tiny,
  numbersep=2pt,
  firstnumber=1,
  numberfirstline=true,
  xrightmargin=0pt,%
  framesep=0pt,%
}

\makeatletter
\gdef\lst@numberfirstlinefalse{\global\let\lst@ifnumberfirstline\iffalse}
\lst@AddToHook{Init}{\global\let\lst@ifnumberfirstline\iftrue}
\makeatother

\lstdefinestyle{nonumbers}{%
  numbers=none,
  xleftmargin=0pt,
}

\lstdefinestyle{inlined}{%
  basicstyle=\lstinlinesize\lstbasicstyle,%
  breakatwhitespace,%
}

\definecolor{darkcyan}{rgb}{0.0, 0.55, 0.55}
\catcode`~=11 
\newcommand{\mytilde}{
    \texttt{~
  }}
\catcode`~=13 

\lstdefinelanguage{meth}[]{Java}{%
  inputencoding=utf8,
  literate=%
  {¬}{{\(\mybld\neg\)}}1%
  {∨}{{\(\mybldbin\vee\)}}1%
  {∧}{{\(\mybldbin\wedge\)}}1%
  {⊔}{{\(\mybldbin\sqcup\)}}1%
  {⊓}{{\(\mybldbin\sqcap\)}}1%
  {⊑}{{\(\mybldbin\sqsubseteq\)}}1%
  {⊥}{{\(\mybld\bot\)}}1%
  {⊤}{{\(\mybld\top\)}}1%
  {~}{{\mytilde}}1%
  ,
  classoffset=1,
  morekeywords={output,checkpoint},keywordstyle=\bfseries\color{red!60!black},
  classoffset=0,
  emph={this},emphstyle={},
  tabsize=4
}

\lstdefinelanguage{secsum}[]{meth}{%
  inputencoding=utf8,
  tabsize=4
}

\newcommand\Summaries[1]{\ensuremath{\Z_{#1}}\xspace}%
\newcommand\Summary[2]{\ensuremath{\mbox{\textsc{Summary}}\lambdatwo{#1}{#2}}\xspace}
\newcommandx\SecuritySemantics[2][1=\hd]{\ensuremath{\mbox{\textsc{Semantics}}^{#1}\lambdaone{#2}}\xspace}
\newcommandx\Coreach[1]{\ensuremath{\mbox{\textsc{Coreach}}\lambdaone{#1}}\xspace}
\newcommandx\ComputeGuard[2]{\ensuremath{\mbox{\textsc{SynthesizeGuard}}\lambdatwo{#1}{#2}}\xspace}
\newcommandx\InferSummary[2]{\ensuremath{\mbox{\textsc{InferSummary}}\lambdatwo{#1}{#2}}\xspace}
\newcommandx\Minimize[2]{\ensuremath{\mbox{\textsc{Minimize}}\lambdatwo{#1}{#2}}\xspace}
\newcommandx\Triang[3]{\ensuremath{\mbox{\textsc{Triang}}\lambdathree{#1}{#2}{#3}}\xspace}
\newcommandx\cofactor[2]{\ensuremath{\mathrm{cofactor}\lambdatwo{#1}{#2}}\xspace}

\renewcommand\Support[1]{\ensuremath{\mathsf{Supp}_{#1}}\xspace}%
\newcommand\Footprint[1]{\ensuremath{\mathsf{Foot}_{#1}}\xspace}%
\newcommand\MethSum[1]{\ensuremath{\mathit{Sum}_{#1}}\xspace}%
\newcommand\SumSupport[1]{\Support{#1}}%
\newcommand\SumFootprint[1]{\Footprint{#1}}%
\newcommand\SumGuard[1]{\ensuremath{\mathit{Guard}_{#1}}\xspace}%
\newcommand\SumEffect[1]{\ensuremath{\mathit{Effect}_{#1}}\xspace}%
\newcommand\havar{ \ensuremath{{\Varset}_{\hvar}} \xspace}%
\newcommand\hfvar[1]{\ensuremath{\Varset_{\retVar {\hvar{}}}}\xspace}%
\newcommand\HDom[1]{\ensuremath{\mathbb{H}
  }\xspace}%
\newcommand\HFoot[1]{\ensuremath{\mathbb{H}_{\mathit{ret}}%
  }\xspace}%

\ifacmartloaded

\else                           

\fi

\ifllncsloaded
\AtBeginDocument{%
}
\fi

\newcommand{\green}[1]{#1}%

\newcommand{\red}[1]{{\color{red}#1}}

\newcommand{\nbrem}[2][]{%
  \ifthenelse{\equal{#1}{yep}}\relax{{\leavevmode
      #2}}}%

\newcommand\Symmaries{\textsf{Symmaries}\xspace}

\ifacmartloaded
\setcopyright{none}
\acmYear{2025}
\acmDOI{10.1145/1122445.1122456}
 \setcopyright{none}
 \setcopyright{acmcopyright}
 \setcopyright{acmlicensed}
 \setcopyright{rightsretained}
 \copyrightyear{2024}
\fi
\excludecomment{leaveout}
\excludecomment{forappendix}

\begin{document}
\title{\Symmaries: Automatic Inference of Formal Security Summaries for Java Programs}

\ifacmartloaded
\author{Narges Khakpour}
\affiliation{%
	\institution{Newcastle University}
	\country{UK}
}
\email{narges.khakpour@ncl.ac.uk}

\author{Nicolas Berthier}
\affiliation{%
	\institution{OcamlPro, France and University of Liverpool}
	\country{France}
}
\orcid{0000-0002-0933-8193}
\email{nicolas.berthier@ocamlpro.com}
\fi

\ifieetranloaded
\author{%
	Narges Khakpour\thanks{Newcastle University, UK and Linn\ae us University, Sweden. Email: narges.khakpour@ncl.ac.uk}%
	\and \hspace{2cm}
	Nicolas Berthier\thanks{OcamlPro, France and University of Liverpool. Email: nicolas.berthier@ocamlpro.com}%
}
\fi

\ifllncsloaded
\author{%
	Narges Khakpour\inst{1,2} \and
	Nicolas Berthier\inst{3,4}%
}
\institute{%
	\inst{1} Newcastle University, UK
	\and
	\inst{2} Linn\ae us University, Sweden
	\and
	\inst{3} OcamlPro, France
	\and
	\inst{4} University of Liverpool, UK
}
\fi


	\begin{abstract}
	We introduce a scalable, modular, and sound approach for automatically constructing {formal security specifications} for Java bytecode programs in the form of method summaries. A summary provides an abstract representation of a method’s security behavior, consisting of the conditions under which the method can be securely invoked, together with specifications of information flows and aliasing updates. Such summaries can be consumed by static code analysis tools and also help developers understand the behavior of code segments, such as libraries, in order to evaluate their security implications when reused in applications.
Our approach is implemented in a tool called \Symmaries, which automates the generation of security summaries. We applied \Symmaries to Java API libraries to extract their security specifications and to large real-world applications to evaluate its scalability. Our results show that the tool successfully scales to analyze applications with hundreds of thousands of lines of code, and
that \Symmaries achieves a promising precision depending on the heap model used. We prove the soundness of our approach in terms of guaranteeing termination-insensitive non-interference.
	\end{abstract}

	\maketitle

\newcommand*\smref[1]{
	\ref{#1}
}


%
%




\section{Introduction}
\label{sec:introduction}
Due to the widespread use of open-source software development as an established practice, vulnerabilities can quickly propagate, especially through common third-party packages that are {directly} or {indirectly} used by many developers.
Approaches to ensuring the security of open-source code focus on various aspects including detecting insecure code~\cite{uss/Xiao0YXYLL0HZS20,sp/KimWLO17,uss/WooCLO23,ccs/DuanBXKL17,icse/WooPKLO21}, isolating malicious code~\cite{ZimmermannSTP19,StaicuPL18}, identifying instances of reused code, whether outdated or modified, through software composition analysis~\cite{icse/WooPKLO21}, and developing trusted package managers to ensure the integrity of uploaded packages~\cite{TorresAriasAKC19,KuppusamyDC17}.
These methods for vulnerability identification often rely on analyzing code patterns and detecting potential security issues through advanced static analysis techniques, e.g., matching code against known reported vulnerabilities. While these efforts have led to scalable approaches that help prevent the use of insecure open-source code, they still fail to detect all vulnerabilities. This limitation stems from the lack of formal foundations necessary to provide sound security guarantees.

\newcommand\FileCacheImageOutputStream{\texttt{FileCacheImageOutputStream} }
\newcommand\Write{\texttt{write} }

This lack of formal foundations motivates the need for methods that can correctly specify code behavior, enabling consumers to rigorously assess and reason about security in a specific context. Since security is defined relative to particular requirements, a code segment that is secure in one setting may be insecure in another. For example, the method \texttt{FileImageOutputStream} from the class \texttt{javax.imageio.stream}, shown in Figure~\ref{fig::running.example.code}, writes a sequence of bytes to the stream at the current position. Internally, it calls \texttt{write} from \texttt{java.io.RandomAccessFile}, which in turn writes the data via the native method \texttt{writeBytes}. The security of this method depends on the sensitivity of the data (\cv{r1}) being written and the stream indices (\cv{i0} and \cv{i1}) it accesses; if used in a context where any of these are highly sensitive, the method could result in information leakage. Therefore, providing formal and correct specifications of method behavior is essential for understanding its security implications.
Manual construction of correct and complete specifications is challenging and error-prone~\cite{ramanathan2007static}. To address this, several methods have been proposed to automatically infer specifications from code~\cite{wang2024dainfer,arzt2016stubdroid,kan2025spectre,eberhardt2019unsupervised,tileria2024docflow,icse/StaicuTSMP20,clapp2014mining}. However, existing approaches suffer from at least one of the challenges C1–C4, which limits their effectiveness in security analysis.
The first challenge is that only a few approaches focus on inferring security-relevant specifications, mainly taint specifications~\cite{tileria2024docflow,icse/StaicuTSMP20}, and these methods are unsound (C2), while soundness is essential for generating formal specifications that correctly capture the security behavior of code.
Moreover, sound formal approaches often face scalability issues when applied to realistic, large-scale codebases (C3). Moreover, the outputs of existing specification inference techniques are mostly designed for static analysis tools and often lack simple, clear and user-friendly conditions that would enable developers to manually assess the security of third-party components for reuse or to evaluate the security of their own code (C4).




\begin{figure}
\lstdefinestyle{numbered}{
  stepnumber=1,
  numbersep=5pt,
  basicstyle=\ttfamily\tiny,
  breaklines=true,
  frame=tb,
moredelim=**[is][\color{white}]{@}{@}**
}
\begin{minipage}{0.31\linewidth}
\begin{lstlisting}[style=numbered,language=Java,escapeinside={(*}{*)}]
class FileImageOutputStream {
 long streamPos;
 java.io.RandomAccessFile raf;
...
void write(byte[] r1,int i0,
            int i1){
  java.io.RandomAccessFile r2;
  long v3, v4;
  flushBits();
  r2=this.raf;
  r2.write(r1, i0, i1)(*\tikzmark{call}*);
  v3=streamPos;
  v4=(v3+ i1);
  streamPos=v4;
  }
}
\end{lstlisting}
\end{minipage}\hfill
\begin{minipage}{0.25\linewidth}
\begin{lstlisting}[style=numbered,language=Java,escapeinside={(*}{*)}]
class RandomAccessFile {

 void (*\tikzmark{def}*)write(byte[] r1,
            int i0,int i1) {
    writeBytes(r1, i0, i1);
    return;
    }
 . . .
}
\end{lstlisting}
\end{minipage}\hfill
\begin{minipage}{0.43\linewidth}
\begin{lstlisting}[mathescape=true,language=meth,basicstyle=\tiny\ttfamily]
void java.io.RandomAccessFile:write
     (byte[] r1,int i0,int i1) {
  // --- Guard ---
  guard := (pc $\sqcup$ level_i0 $\sqcup$ level_i1 $\sqcup$ level_r1 $\sqcup$ level_this
        $\sqcup$ objlevel_r1 $\sqcup$ objlevel_this = $\low$);
  // --- Updates to security levels ---
  objlevel_r1   := $\low$;
  objlevel_this := $\low$;

  // --- Updates to aliasing specification
  falias_this_r1  := falias_this_r1;
  falias_this_this:=falias_this_r1 $\lor$       
                    falias_this_this;
}
\end{lstlisting}
\end{minipage}\hfill

\tikz[overlay,remember picture]
  \draw[->,thick,dotted,black] (call) .. controls +(1,-1) and +(-1,1) .. (def);

\caption{{An example Java bytecode (left) and  the security summary for \texttt{RandomAccessFile.}\Write} (right)}
\label{fig::running.example.code}
\end{figure}



\red{
}


To address this gap, we propose an approach to automatically construct formal specifications for Java methods at the bytecode level that captures two main aspects. The first aspect concerns the security-related behavior of the method in terms of \emph{information flow control} (IFC)~\cite{sabelfeld2003language}, a well-established technique for analyzing application security that prevents unauthorized disclosure or modification of sensitive data. IFC is often formalized using noninterference-based properties~\cite{GuoguenMeseguer1982SecPolsNSecMods}, which ensure that an attacker's observations cannot be influenced by high-sensitivity data. Our specifications track information flows via heap and consider both explicit and implicit flows.
The second aspect concerns the method's aliasing specification, which describes how method calls modify the heap specified as relationships between references, e.g., tracking which objects may point to the same memory location (may-alias relations) and how these relationships change after executing a method. By capturing updates to the heap and its aliasing structure, our approach enables a correct understanding of how information flows via the heap, and it helps understand and reason about side effects, potential information leaks, and the overall security impact of a method.

We develop \Symmaries, a \emph{scalable} tool that employs a formal \emph{summary-based} approach to construct \emph{sound security summaries} for Java methods. We model program semantics as a symbolic control flow graph that assigns security levels to variables and references and represents the heap as a set of aliasing relationships among references. The semantics captures how these security levels and heap aliasing relationships evolve as a result of a method's behavior.
A method summary consists of two main components: (i) a sound \emph{guard}, which specifies the conditions under which it is secure to invoke the method, expressed as a concise logical expression; and (ii) an \emph{update}, which describes the effects of the method invocation on the program state, represented as a set of assignments. The guard and update are defined over aliasing relationships in the heap and the security levels of the method's arguments and \texttt{this}.
\figurename~\ref{fig::running.example.code} presents an example security summary constructed for \Write, where a variable \cv{falias\_r1\_r2} denotes the field-aliasing relationship between two typical references \cv{r1} and \cv{r2} (i.e.,
that \cv{r1} may alias \cv{r2} through a sequence of field accesses), and \cv{objectLevel\_r} and \cv{level\_r} indicate the security levels of reference \cv{r} (See Section~\ref{sec:input-programs} for more details).
The guard of this summary states that this method is secure to call in any context where none of its arguments or \this contain high-sensitive information.
We believe that our security summaries provide a \emph{meaningful}, \emph{sound}, and \emph{intelligible} abstraction of a method's security behavior. They can be leveraged by static analysis tools or by code consumers to reason about the security of code without inspecting implementation details.


\subsubsection*{Summary of Contributions}
We propose a scalable, precise, and sound approach for security summary construction of Java bytecode programs. Our contributions are as follows:

\begin{itemize}[nosep,left=0pt]
  \item We propose a new method for constructing security summaries of Java (bytecode) programs, providing an intelligible abstraction of their security behavior. Our approach modularly analyzes programs to extract their summaries, captures information flows via the heap, and considers both explicit and implicit flows.

  \item We implement \texttt{\Symmaries}, a complete toolchain supporting this method, and evaluate its scalability through extensive experiments.
  Our experiments demonstrate that our summary-based approach scales to large, real-world open-source applications.
  The tool produces a set of method contracts as output, which we designed to be easily understandable by users.

  \item Owing to its symbolic design, \texttt{\Symmaries} requires no manual annotations (e.g., entry point identification), making it practical and lightweight to apply to large applications, as confirmed by our evaluation.

  \item The precision of inferred security summaries depends on several sensitivity dimensions, including context, object, flow, and field. Our approach is flow-sensitive, context-sensitive, and object-sensitive, while field-sensitivity depends on the chosen heap model. \texttt{\Symmaries} provides multiple heap models in which the heap is represented as different aliasing relations ~\cite{BerthierKhakpour23Heap}. Hence, the precision of results depends on the selected heap model.
  {We applied \Symmaries to a subset of the Java JDK to evaluate precision, and our results show that the tool infers precise specifications for methods.}

  \item We formally prove the soundness of our approach with respect to termination-insensitive non-interference.
\end{itemize}

\begin{leaveout}



"Very good intro to the problem by V1SCAN"



\todo{Nice intro for the need to summaries specification and selling the story "Extracting Taint Specifications for JavaScript Libraries.pdf"
}

\begin{redenv}
We first give in Section~\ref{sec:synth-symb-summ}
on the kind of symbolic models and the related controller synthesis
algorithm for safety that we use.
Next, we start Section~\ref{sec:input-programs-analysis-framework}
with a description of syntax and concrete semantics of the input
programs we consider{, and then turn to details about the
  parametric heap abstract domain and the structure of symbolic
  summaries}.

\begingroup\itshape%
\begin{itemize}[nosep,leftmargin=1em]
\item 
  We first lay out the requirements for symbolic heap abstractions
  that enable the capture of information flows within the heap and define
  symbolic summaries (Section~\ref{sec:symbolic-summaries-n-heap-abstraction});
\item We then develop a technique for symbolically modeling the
  security semantics of low-level code in a way that enables one to
  capture implicit flows
  (Section~\ref{sec:symbolic-encoding-security-semantics});
\item 
  We propose the first fully automated 
  and scalable field-insensitive summary-based analysis tool that addresses IFC, called \Symmaries, and
  give a proof of termination-insensitive noninterference for
  our 
  approach (Section~\ref{sec:interprocedural});
\item We detail our implementation in a complete tool-chain, and
  exercise it on \IFSPEC benchmarks~\cite{HamannHMM0T18}.
  We obtain significant improvements on precision over the
  state-of-the-art whilst retaining soundness.
  We also demonstrate the scalability of our approach by reporting on
  our extensive analysis of {70} real-life web-applications from the
  ABM benchmark~\cite{do2016toward}
  (Section~\ref{sec:experimental-results}).
\end{itemize}
\endgroup
\itodo{Link to what's in SM already…}
\end{redenv}

\end{leaveout}

\renewcommand\figurename{Fig.}
\section{Basic Concepts}

\subsubsection*{Security Domain}
We follow a type-based approach to formalize security where
security types are assigned to program variables and objects. We consider a
classical two-level \emph{security domain} that is defined as a
bounded lattice \(\langle \LL, \sqsupseteq 
\rangle \), where 
\(\LL ⊇ ｛\low, \high｝\) is the finite set of security levels, \(⊑\) is a
partial order defined over \LL{, and \(\low\) and \(\high\) respectively
	denote the least (low) and greatest (high) security levels}.
We will typically define a variable \plvl x to denote the
security level associated with a program variable \(x\) (denoted by \cv{level\_x} in \figurename~\ref{fig::running.example.code}).
We then focus on enforcing a confidentiality policy where
high-sensitive (\high) information should not influence low-sensitive (\low)
outputs.

\subsubsection*{Input Programs}
\label{sec:input-programs}
We consider input programs as a set of methods that manipulate primitive data, and references to objects.
To maximize the ready availability of our approach
, we concentrate on a low-level language where 
most statements 
correspond to %
a family of \Java bytecodes, in a way similar to
\textsf{Jimple}~\citep{vallee2010soot}.
The \emph{body}
of a method \(m\)
is a non-empty finite
semicolon-separated sequence of statements, built according to \stms
in the following grammar:
{\setlength\abovedisplayskip{3pt}\setlength\belowdisplayskip{3pt}%
	\[
	\begin{array}{r@{\ }r@{\ }l}
		\stms &⩴& [\lbl{l}:]~\stm;\stms \sep \surd \\
		\stm &⩴& {\ipstms} \sep [v = ] r.\meth(\lits) \sep \coutput l (\elts)  \sep \creturn~ [v \sep r]
		\\
		\ipstms &⩴& v = e \sep v = r.f \sep r.f = e  \sep r = r \sep r = \cnew~c \sep \\
		&& r = \Null  \sep  [\cif~(e)]~\cgoto~\lbl{l} \\
	\end{array}
	\]}%
where square brackets denote optional constructs,
\(\surd\) is an empty sequence of statements, and \(\lbl{l}\) is a label that uniquely identifies a statement.
\(r.\meth(\lits)\) is a call to the method \meth of the object
referenced by \(r\), with a mapping \lits that associates a
\emph{value} with each formal parameter of \meth.
Our language additionally features a statement \({\coutput
	l(\elts)}\)
that outputs \elts
to a public channel. 
\red{\ipstms} shows the syntax for intra-procedural  statements, e.g., assignments, branches, and constructors, where $e$ is an expression, $v$ is a variable and $f$ is a field.

\subsubsection*{Heap Model}
\label{sec:heap-abstraction}
We use a {heap abstraction} \GeneralH to capture information flow to and from heap objects
following the approach introduced in \cite{BerthierKhakpour23Heap}, which maintains,
\begin{itemize}[nosep,left=0pt]
	\item a \emph{heap typing environment} that
	associates an upper-bound on the security level of the portion of heap
	that is reachable via any reference \(r ∈ R\) where  $R$ is a set of references.  The security level of an object referenced by $r$ is denoted by \hlvl r.

	\item (\emph{over-approximations} of) \emph{heap-related relations} between the  reference
	variables in \(R\) and/or the objects that are reachable thereof.
	An example of a heap relation is the \emph{may-alias} relation that holds between two reference variables, if they both point to the same object.
\end{itemize}
\setlength{\intextsep}{5pt} 
\setlength{\columnsep}{10pt} 
\begin{figure}
  \centering
  \begin{tikzpicture}[node distance = 6mm]
    \def\smallsp{1.4em}%
    \def\medsp{7em}%
    \node [draw, rounded corners] (a) {$a:\high$};
    \node [draw, rounded corners, right = of a] (b) {$b:\high$};
    \node [draw, rounded corners, below = of  a] (c) {$c:\low$};
    \node [draw, rounded corners, right = \medsp of b] (d) {$d:\high$};
    \node [draw, rounded corners, below = of d] (f) {$e:\low$};

    \draw (a.east) to node [above, anchor=south east] {$\AliasRelName$} (b);
    \draw (a) to node [left] {$\AliasRelName$} (c);
    \draw (b) to node [right, anchor=north west] {$\AliasRelName$} (c);
    \draw (d) to node [above, anchor=east] {$\AliasRelName$} (f);
  \end{tikzpicture}
  \vspace{1mm}
  \hspace{-3.5cm}
  \parbox{\linewidth}{ 
  \[
  \begin{array}{ll}
    \GeneralH = &
    (\hlvl a = \high) \wedge
    (\hlvl b = \high) \wedge
    (\hlvl c = \low) \wedge
    (\hlvl d = \high) \wedge
    (\hlvl e = \low) \wedge\\
    &
    \AliasRel a b \wedge
    \AliasRel a c \wedge
    \AliasRel b c \wedge
    \AliasRel d e
  \end{array}
  \]
  }
  \caption{A heap model and its formal representation \(\GeneralH\)}
  \label{fig:heap-model}
\end{figure}
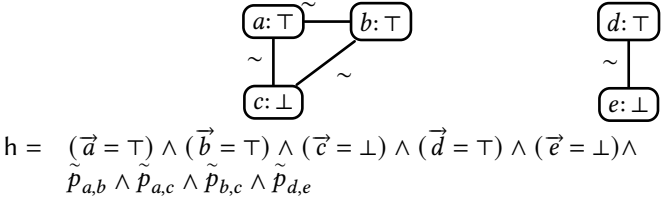
We formalize a heap as a predicate over a set of variables representing the heap typing environment and the heap-related relations. A relation is formally specified as a predicate over a set of Boolean variables, that each indicates whether two references are in the relation or not (\ie we use predicate abstraction where a Boolean variable specifies whether a relation between two references holds).
\figurename~\ref{fig:heap-model} shows an example heap and its formal representation.
The nodes show the reference variables and their security types, and arcs between two references indicates they are in a \emph{may-share} heap-related relation: for any given \(｛r, s｝⊆ R\), \(r\) and \(s\) may share iff there may exist an object that is (transitively) reachable via both \(r\) and
\(s\).
In this formal representation,
the variable \ensuremath{\hlvl[]r} (i.e., correspond to the variable \cv{objlevel\_r} in \figurename~\ref{fig::running.example.code}) contains the \emph{upper bound} on
	the security levels of any object that is reachable via \(r ∈ R\),
and a proposition \ensuremath{\GenericRel[] r s } holds whenever \(r\) and \(s\) are in the may-share relation.
\begingroup%
\def\fvar{\hvar}%
This formula states that the references $a$, $b$, and $d$ contain high-sensitive information (\eg $\hlvl a = \high$), while $e$ and $c$ are low-sensitive (\eg $\hlvl e = \low$). Furthermore, the references associated with the edges may alias each other; for example, $a$ and $c$ may alias each other (\eg the propositional variable $\GenericRel a c$ holds).
 Note that different types of relations can be used to represent the heap.
 For instance, we use \TransFieldAliasRel r s  to represent the field-aliasing relations (i.e., the variable \cv{falias\_x\_y} in \figurename~\ref{fig::running.example.code}).
See~\cite{BerthierKhakpour23Heap} for some examples of possible relations.


\subsubsection*{Symbolic Control-flow Graphs --- SCFGs}
The transition systems that we use to specify the program semantics are 
traditional labeled transition systems augmented with sets of \emph{state} and \emph{input} variables, respectively denoted \(X\) and \(I\).
The values for input variables can be seen as coming from the environment of the system.
Figure~\ref{fig:overview} shows an SCFG describing the semantics of the method  \cv{FileImageOutputStream.}\Write.
Predicates and right-hand-side expressions in assignments are built using traditional logical connectives (\ie ¬, ∨, ∧ and ⇒), along with a ternary conditional construct ``\ite{⋅}{⋅}{⋅}'' with an obvious meaning. 
The symbolic variables we make use of typically take their value in the security domain \LL, or the set of Booleans \(\BB ≝ ｛\ff, \tt｝\)
.
We use a \emph{merge operation} \SQMergeiu to merge variable assignments
.
This operation is obtained as a union where multiple expressions assigned to a variable \(v\) are combined using some 
connective \(⊔_v\).
The latter depends on the semantics of each variable:
as we only use variables to hold over-approximations in our encoding, we use the dis\-junction ∨ for Booleans, and the least-upper-bound ⊔ for security levels.
For instance, \(\{a \assign \tt, b \assign \ff\} \SQMergeiu \{b\assign \tt\} = \{a \assign \tt, b \assign \ff ∨ \tt\}\).
\newcommand{\Sf}{\ensuremath{S}\xspace}%
\green
{The semantics of \Sf  is a finite-state automaton  \FSM\Sf=\FSMdef whose \emph{state-space} \(\Q\) is the Cartesian product of the set of locations Λ and the set of all possible valuations for the state variables, \ie \Val X.
	\(\Q_0 \subseteq \Q\) is the set of initial valuations of state variables that satisfy $X_0$.
	\FSM\Sf takes one transition from \smTrans whenever it receives a valuation for \emph{all} the input variables $\mathcal{I}$, \ie an element in \Val I.
	In any location, there is always exactly one transition whose guard is satisfied by the valuations for all the variables.
	When this transition is taken,
	its assignments are applied to update the state variables.}
An \emph{invariant} \(φ 
\) 
for the SCFG \(S\) is a mapping from locations to predicates on state variables.
\(S\) \emph{satisfies} \(φ\) iff every state \(
q
\) with location ℓ that is reachable by S is such that~\(
q
⊨ φ
(ℓ)
\).


\section{Method Overview}
\begin{figure}
	\centering
\input{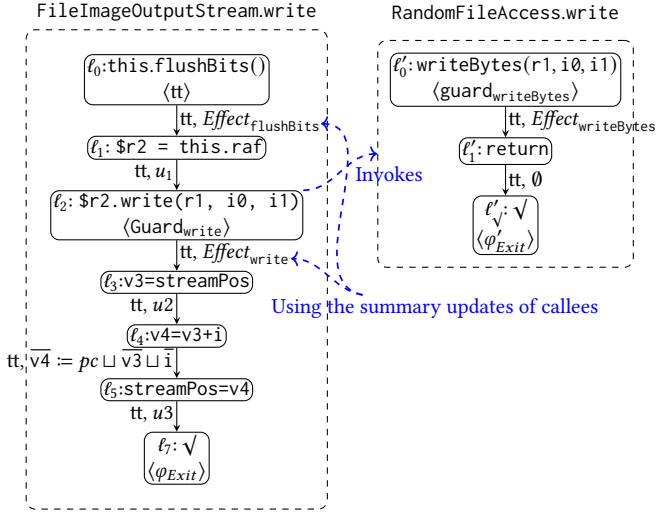}%
\newcommand{\flushBits}{\cv{flushBits}}
\scalebox{0.8}{
\begin{tikzpicture}[thick, node distance=5mm,
  basic-arrow/.style={->, >=stealth},
  analysis-arrow/.style={->, thick, dashed, blue}]

\node [draw, rounded corners] (a0) {$ℓ_0:$\cv{this.flushBits()}\\ $\langle \tt \rangle$};
\node [draw, rounded corners, below=of a0] (a1) {$ℓ_1:$ \texttt{\$r2 = this.raf}};
\node [draw, rounded corners, below=of a1] (a2) {$ℓ_2:$ \texttt{\$r2.write(r1, i0, i1)}\\ $\langle \mathtt{Guard}_{\cv{write}} \rangle$};
\node [rounded corners, below=of a2,draw] (a3) {$ℓ_3:$\texttt{v3=streamPos}};
\node [draw, rounded corners, below=of a3] (a4) {$ℓ_4:$\texttt{v4=v3+i}};
\node [draw, rounded corners, below=of a4] (a5) {$ℓ_5:$\texttt{streamPos=v4}};
\node [draw, rounded corners, below=of a5] (sink) {$ℓ_7:\surd$\\ $\langle φ_{Exit}\rangle$};

\draw[basic-arrow] (a0) to node [right] {\tt, \(\SumEffect{\flushBits}\)} (a1);
\draw[basic-arrow] (a1) to node [left, pos=.4] {\(\tt, u_1\)} (a2);
\draw[basic-arrow] (a2) to node [right, pos=.4] {\strut\tt, \SumEffect{\cv{write}}} (a3);
\draw[basic-arrow] (a3) to node [left, pos=.4] {\tt, \(u2\)} (a4);
\draw[basic-arrow] (a4) to node [left, pos=.4] {\tt, \(\cplvl{v4} \assign \pc ⊔ \cplvl{v3} ⊔ \cplvl{i}\)} (a5);
\draw[basic-arrow] (a5) to node [left, pos=.4] {\tt,  \(u3\)} (sink);

\node [fit={(a0) (a1) (a2) (a4) (a5) (sink)}, draw, dashed, inner sep=4mm, rounded corners=1ex] (scfg-main-box) {};

\node [draw, rounded corners, right=2cm of a0] (fa1) {\(ℓ'_0: {\cv{writeBytes (r1, i0, i1)}}\) \\  \(\langle \cv{guard}_{\cv{writeBytes}}\rangle \)};
\node [draw, rounded corners, below=of fa1] (fa2) {\(ℓ'_1: \creturn\)};
\node [draw, rounded corners, below=of fa2] (fsink) {$ℓ'_\surd:\surd$\\ $\langle φ'_{Exit}\rangle$};

\draw[basic-arrow] (fa1) to node [right] {\tt, \(\SumEffect{\cv{writeBytes}}\)} (fa2);
\draw[basic-arrow] (fa2) to node [right] {\tt, \(\emptyset\)}  (fsink);

\node [fit={(fa1) (fa2) (fsink)}, draw, dashed, inner sep=2mm, rounded corners=1ex] (scfg-write-box) {};

\draw[analysis-arrow] (a2.north east) to [out=0, in=180] node [above, xshift=2mm, yshift=1.5mm, anchor=north west] {Invokes} (scfg-write-box);

\node [rounded corners, below left=of fsink, xshift=25mm, yshift=-3mm] (use-box) {\color{blue}{Using the summary updates of callees}};
\foreach \i/\x/\y in {a2/19/3, a0/24/3} {
    \node [below=of \i, xshift=\x mm, yshift=\y mm] (dummy\i) {};
    \draw[analysis-arrow] (use-box) to [out=170, in=0] (dummy\i.south);
}

\node [fit={(scfg-write-box) (scfg-main-box)}, inner sep=0, outer sep=0] (scfgs-bbox) {};

\node [above=of a0] {\cv{FileImageOutputStream.write}};
\node [above=of fa1] {\cv{RandomFileAccess.write}};

\end{tikzpicture}
}

\hrule width 0pt              
 
	\caption{Our Summary-Based Approach}
	\label{fig:overview}
\end{figure}

We design a \emph{sound} and \emph{scalable} approach for constructing information flow summaries of methods in a modular way, where each method is analysed individually.
We assign \emph{security levels} to variables and heap objects manipulated by a method,
and specify the program semantics using symbolic control flow graphs.
The state variables of program semantics include (i) the security typing environment (\ie the security
level of variables and the level of context denoted \pc---for ``program context''), as well as (ii) the {heap model} used to capture information flows to and from heap objects.
{We follow the approach presented in ~\cite{BerthierKhakpour23Heap} to specify the  semantics of non-method call statements}.
To define the  semantics of method invocations, we follow a summary-based approach and re\-use the summary for the invoked method instead of inlining the semantics of its body into the caller's semantics.
The summary of a method is a \emph{contract} that consists of
\begin{itemize}[nosep,left=0pt]
	\item a \emph{guard} to specify the invocation conditions under which the method call is secure, \ie there is no illegal information flow in the method. A guard is described as constraints on the security levels  and heap structure of the method arguments.
	\item an \emph{effect} to represent its side-effects on security levels and aliasing relations among the heap objects, i.e., it is in principle a set of assignments describing how the aliasing relations and security labels are updated by the method (See the example in \figurename~\ref{fig::running.example.code}).
\end{itemize}
%
\def\labelloop{end}%
\def\labelloopcond{loop}%
\def\sh{\mathtt{s}}%
\def\eh{\mathtt{e}}%
\def\ih{\mathtt{b}}%
\def\thish{\mathtt{this}}%
\def\fix{\mathtt{fix}}%
\def\source{\mathtt{source}}%
\colorlet{rho3color}{green!46!black}%
\colorlet{rho5color}{red!74!black}%
\def\Cbar{\ensuremath{\mathtt{bar}}\xspace}%
\def\Cfix{\ensuremath{\mathtt{fix}}\xspace}%
\def\Csource{\ensuremath{\mathtt{source}}\xspace}%
\def\Sbar{\ensuremath{S_\Cbar}\xspace}%

\newcommand{\flushBits}{\cv{flushBits}}

\figurename~\ref{fig:overview} presents an SCFG that specifies the simplified semantics of the program shown in Figure~\ref{fig::running.example.code}, where locations correspond to statements and transitions update both security levels and {aliasing relations}.
Each location label represents a statement along with a possible invariant, if any, which shows the specification that must hold at that point.
Each transition is labeled with a guard and an update function.
To define a method's semantics, we use the summary update of a callee in place of its invocation. For example, the transition from $\ell_0$ in \cv{FileImageOutputStream.write} updates the state variables according to \flushBits's effect, i.e., \SumEffect\flushBits.
We also use the summary guard of the callee as a safety property (referred to as an invariant) at the invocation location to prevent any illegal information flow, e.g., $\mathtt{Guard}_{\cv{write}}$ at location \({ℓ_2}\) or true in location $\ell_0$.
Further, \cv{u1-u3} are a set of updates defined according to the semantics of field load and store statements.

To construct method summaries, {we find all path conditions that avoid reaching \emph{specific} states using a co-reachability analysis.}
To build the summary guard, we find all the path conditions that do not lead to \emph{insecure} states, \ie a state is insecure if it's either a sink statement leaking high-sensitive information, or is a method invocation reaching an insecure state directly or indirectly. For instance, the method \Write in Figure~\ref{fig:overview}: (i) should not leak any sensitive information in the sink statement associated with location $\ell_7$, and (ii) the invocations to the methods \flushBits and \cv{RandomFileAccess.}\Write should not result in reaching an insecure state. To ensure case (i), we require the invariant $\invars_{Exit}$ to hold at location $\ell_4$,
which informally states that the security levels of the arguments, \this, and the variables modeling the heap remain unchanged.
This encoding allows us to identify conditions that violate this invariant, which are then used to extract the updates.
To ensure case (ii), we require the summary guard of  \flushBits and  and \texttt{RandomFileAccess.}\Write to hold when they are invoked (defined as the true invariant $\tt$ in $\ell_0$ and $\mathtt{Guard}_{\cv{write}}$ in $\ell_2$, \ie no insecure state will be reached via calling these methods.

	To build summary effects, we find all path conditions leading to an update in the heap or security levels of the arguments.
	In our example, we store the initial heap $\hvar$ in $\retVar{\hvar{}}$ upon method entry, and never modify $\retVar{\hvar{}}$ during execution. We then use the invariant \(φ_{Exit}\) in \({ℓ_7}\), which is the method exit point where \(\hvar\) is the heap updated by the method, to express {symbolic upper-bounds} on the security types or heap relations among references upon exit.
	{Our co-reachability analysis then identifies all the path conditions $\phi$ that update the heap, which are the conditions that violate the symbolic bounds in \(φ_{Exit}\), as well as the path conditions that lead to insecure states, which are the conditions that violate the other invariants.}
	We use a \emph{triangularization} technique~\cite{Hietter-triang} to transform $\phi$ into a guard and a set of assignments that constitute the summary for the method.
We prove that summaries constructed using our approach are sound (See Section~\ref{sec:full-bottom-up}),
which enables us to reuse them in place of method invocations.
{The expressive power of summaries depends on the chosen heap-related relations (See~\cite{BerthierKhakpour23Heap} for examples of heap abstractions with different heap-related relations). Leaving this set of relations largely unspecified below allows us to abstract away unnecessary details in our exposition, where we focus on binary relations without loss of generality.}

For the inter-procedural analysis, we take a modular summary-based approach and analyse each method separately to construct its summary, thereby significantly enhancing the scalability of our analysis.
If the program under scrutiny does not feature (potentially) recursive
methods (\ie its call graph \(G\)  is a\-cyclic), one can use a simple
topological traversal of \(G\) to compute every summary.
Yet, this solution is not applicable if \(G\) is cyclic, \ie involves
recursive methods. 
{To overcome this issue, we define summaries in such a way that we can follow 
a
\emph{chaotic iteration strategy}~\citep{Bourdoncle:1993:ChaoticIterationStrategiesWithWidenings} to find a fixed-point that gives the set of all summaries for a program.}


\begin{leaveout}

	\begin{figure}%
		\centering%
		\smaller\smaller%
		\caption{Computing the symbolic summary 
			of a method \(m\).}%
		\label{fig:approach-overview-single-method}%
\end{figure}

: we compute
a \emph{controller} {\(K_m\)} that is such that a predefined
combination of \(S_m\) and \(K_m\) satisfies \(φ_m\).
More specifically, we construct \(S_m\) and \(φ_m\) so that \(K_m\)
actually constitutes a symbolic summary \MethSum m for \(m\).
\end{leaveout}
%
\begin{leaveout}
This summary provides a \emph{sound judgment on the circumstances upon
	which \(m\) satisfies the desired security
	property}
: this takes the form of a \emph{predicate} expressed on security
levels and heap-related relations pertaining to \(m\)'s formal
arguments.
Moreover, \MethSum m additionally includes a \emph{transformer} that
finitely represents every possible effect of \(m\) on all its
(infinitely many) possible calling contexts: this notably comprises
the relevant information flows it may induce, and its effects on
heap-related relations.
…
This appears in \figurename~\ref{fig:approach-overview-single-method},
where the translation of \(m\) into \(S_m\) and \(φ_m\) makes use of
summaries \MethSum n and \MethSum o obtained from other methods \(n\)
and \(o\).

\begin{redenv}
	Still, relying on such formal techniques poses a significant threat to
	scalability, and formal tools also tend to provide results in the form
	of large formulas or mathematical objects that are hard to interpret.
	\emph{Summary-based} (\aka bottom-up) static analysis
	approaches 
	come to our rescue at this
	point~\citep{SharirPnueli1978TwoApproaches2InterpDataflowAnalysis,Yorsh2008PreciseAndConciseProcedureSummaries}.
	\nbrem[yep]{\citet{Bodden:2018:SSE:3236454.3236500}, which currently
		drives the group that designs and extends the widely used static
		analysis framework for \Java called \soot~\citep{vallee2010soot},
		strongly advocates using summary-based approaches for the improved
		precision and simpler overall analysis complexities that they can
		bring over whole-program analyses.\todo{Move to RW?}}
	Indeed, whereas a whole-program analysis essentially builds a
	potentially large data-structure which then serves to answer queries
	about the program, a bottom-up approach typically computes a
	\emph{summary} for each procedure, that describes its aspects that are
	relevant to the analysis domain.
	Summaries are reused whenever a call to a procedure is encountered,
	and
the challenge is thus to represent them in a way that eases both their
computation and \emph{re\-usability}.

\end{redenv}
\end{leaveout}

		\vspace{-0.1in}
\section{Summary Construction}
\label{sec:symbolic-summaries}
In this section, we first formalize the concept of summary, and then compute the summary of a method that does not contain any method call.
		\vspace{-0.1in}
\subsection{Definition of  summaries}
Consider  a  method \(m\) with
body \MethBody m and
arguments \MethFormalArgs{m} , which includes the reference variable \cv{this} in the case of a non-static method.
{Let \(\MethRefs m \) be the set of all reference variables used in \(m\), including the reference arguments.}
We assume that the reference arguments of \(m\) never appear on the
left-hand side of any assignment in \MethBody m---this preliminary
requirement is easy to enforce via simple program transformations.
We define the \emph{support} of a method $m$ as a set of variables that consists of: the security types of arguments, the \emph{program context} variable
\pc,
and the set of  variables \havar used to encode the reachable portions of the heap upon method call, \ie

$$\SumSupport m ≝  ⋃\nolimits_{x 
	{∈ \MethFormalArgs m}}｛\plvl x｝%
∪ ｛\pc｝\cup {\havar m}.$$%

The \emph{footprint} of $m$ is defined in a similar way, that includes the security level \plvl{π} of the returned
value $\pi$, if any, and the set of variables used to encode the reachable portions of the heap $\hfvar m$ when the
call terminates, \ie
$\SumFootprint m ≝
{\plvl{π}
}
\cup {{\hfvar m}}.$


The summary guard \SumGuard m only involves the variables from the support \SumSupport m
which holds if no execution of \(m\) leads
to a violation of any invariant on its semantics
.
In turn, the effect \SumEffect m
comprises one assignment for each variable in the
footprint \SumFootprint m, \ie it
assigns each component of the security typing
environment
with an expression that gives an \emph{upper-bound} on the
corresponding security level \emph{upon termination of \(m\)}---and
similarly for any over-approximated heap-related relation.

\vspace{-0.1in}
\subsection{Computing summaries}
\label{sec:meth-preambl-synth}
We summarize the overall procedure to compute summaries 
in \algorithmcfname~\ref{alg:synthesis.algorithm} that we elaborate in this section.

\begin{figure}
	\newcommand\Foot{\ensuremath{F}\xspace}%
	\begin{ruleset}
		\myfreerule{\OutputRule}{%
		}{%
			\semloc{\s{\coutput l (x)};\stms ~}\trans{\tt, ∅}\newloc{~\stms}
		}%
		\hspace{2em}%
		\myfreerule{\SinkRule}{%
		}{%
			\semloc{\surd}%
			\trans{\tt, ∅}~ℓ
		}
		\\[\myrulespace]%
		\myrule{\ReturnRule}{%
			\begin{array}{@{}c@{}}
				T = 
					［\sVar \assign \sSpec x］ \text{~if~} \stm = \s{\creturn~x}\mbox,~
     ∅ \text{~otherwise}
			\end{array}
		}{%
			\semloc{ \s{\creturn~[\_]};\stms~}~\trans{\tt, T}\newloc{\surd}%
		}
	\end{ruleset}
	\smaller
	\setlength{\belowdisplayskip}{-3pt}%
	\begin{align}
		φ_\rname{{Sink}}= & ~\pc ⊑ \low ∧ \plvl x ⊑ \low ∧
			(\hlvl[\hvar] x ⊑ \low \text{~if~} x ∈ \MethRefs m\mbox,
			\tt \text{~otherwise})
		\tag{\normalsize\rname{φ-Sink}}\label{eq:invariant4sink} \\
		φ_\rname{Exit} ≝\,
		&
		φ_x
		∧ φ_{\hvar{}}  \tag{\normalsize\rname{φ-Exit}}\label{eq:invariant4exit} \\
		φ_x ≝\,&
		⋀_{\quad\mathclap{\s{\creturn~x} ∈ \MethBody m}}
		{（\sVar = \sSpec x ⇒ {\pc ⊔} \plvl x ⊑ \retVar{{\plvl{{\pi}}}} 
			）}\tag{\normalsize\rname{φ-RetVal}}\label{eq:invariant4exit.value} \\
		φ_{\hvar{}} ≝\,
		&
		⋀_{\quad\mathclap{\s{\creturn~r} ∈ \MethBody m}}
		{（\sVar = \sSpec r ⇒ \hvar[π←r] ⊑_{\mathsf H} \retVar{\hvar{}}）} \tag{\normalsize\rname{φ-Heap}\label{eq:invariant4exit.value}}
		\\&
		∧ （\hvar ⊑_{\mathsf H} \retVar{\hvar{}} \text{~if~} ∄\s{\creturn~r} ∈ \MethBody m
		\mbox,~\tt \text{~otherwise}）
		\notag
	\end{align}
	\caption{Rules for inferring summaries.}
	\label{fig:method-preambles-n-return}
\end{figure}

%
\begin{algorithm}[t]
\smaller
  \KwIn{Method to analyze \(m\)}
  \KwResult{Polymorphic IF-summary \(\MethSum m = 〈\SumEffect m, \SumGuard m〉\)}
  \Comment{Encode the semantics of \(m\) as an SCFG \(S_m\) where every state variable related to the calling context is \emph{left uninitialized}, and express the security requirement as a predicate \(φ_m(ℓ)\) on state variables for each location \(ℓ\) of \(S_m\):}
  \((S_m, φ_m) ← \SecuritySemantics[]{m}\)\;
  \Comment{Define all known unsafe states, that include every state where {safety properties} are violated ($Λ_m$ is the set of $S_m$'s locations):}
  \(\B_0 ← ｛ℓ ⟼ ¬φ_m(ℓ)｜ℓ∈Λ_m｝\)\;
  \Comment{Compute all states that are co-reachable to \(\B_0\):}
  \(\B_∞ ← \Coreach{S_m,\B_0}\)\;
  \Comment{Factor out the state variables that are not part of the calling context {(so \(G\) only involves variables from \(\SumSupport m ∪ \SumFootprint m\))}:}
  \(G ← \cofactor{¬\B_∞(ℓ_0(S_m))}{X_0(S_m)}\)\;\label{alg-line:cofactor-init}
  \(〈\SumEffect m, \SumGuard m〉 ← \Triang G {\SumFootprint m} {∅}\)\;
  \caption{\InferSummary{}{}\label{alg:summary-synthesis}}
  \label{alg:synthesis.algorithm}
\end{algorithm}

\subsubsection{Semantics Rules and Invariants}
To compute an summary for a method \(m\),
\SecuritySemantics[]{m} in our algorithm constructs an SCFG \(S_m\) that specifies its semantics in addition to an invariant mapping \(φ_m\) for each location.
\figurename~\ref{fig:method-preambles-n-return} introduces the semantics of inter-procedural statements and the invariants.

\begin{leaveout}
{Our technique for capturing the implicit flows of information that are
	induced when a branch subject to a high condition is \emph{not} taken,
	consists in encoding an alternation between \emph{nominal} executions,
	and \emph{upgrade analyses of every possible execution paths within
		CDRs}.
	In nominal mode, the updates to security levels reflect every
	information flow except {indirect implicit} flows,
	\ie only explicit and {direct}
	flows 
	are taken into account.
	During upgrade analyses, however, the updates to security levels only
	capture information flows from the execution context \pc.
	We implement this behavior with the following constructions to
	seamlessly distinguish 
	the nominal executions (when state variable \uVar does not hold) and
	upgrade analyses (\ie \uVar holds):%
	\begingroup%
	\setlength{\abovedisplayskip}{1pt}%
	\setlength{\belowdisplayskip}{3pt}%
	\begin{align}
		\nomlvl l &\ ≝\ \phantom{\plvl x \assign\,}（\ite{\uVar 
		}{⊥}{l}） ⊔ \pc
		\label{eq:level-cap} \\
		\plvl x \assignv l &\ ≝\ \plvl x \assign （\ite{\uVar 
		}{\plvl x}{l}） ⊔ \pc
		\label{eq:assignv}
	\end{align}
	\endgroup%
	In nominal mode, \nomlvl l encodes the least upper-bound between \(l\)
	and the context level \pc, and \(\plvl x \assignv l\) models a
	\emph{strong update} of the security level assigned to \(x\) with
	\nomlvl l.
	In upgrade analysis mode, however, \nomlvl l is equal to the context
	level, and \(\plvl x \assignv l\) encodes a \emph{weak update} of
	\plvl x with \pc.}
\end{leaveout}
To ensure confidentiality, we associate an objective invariant $φ_\rname{{Sink}}$ with every sink statement leading to the application of the \OutputRule rule. This property avoids any breach of confidentiality by ensuring that the security level of any scalar data or referenced object output to a public channel is not high-sensitive.
This property additionally
ensures that an output can only be performed in a low-sensitive context
by asserting \(\pc ⊑ \low\) to avoid leakage via context.
For instance, $φ_{ℓ_4}$ in \figurename~\ref{fig:overview} shows an invariant defined for the sink location ${ℓ_4}$ that outputs $\cv s$. This property requires that the security level of \cv s and the context be low.

The rule \ReturnRule defines the semantics of \creturn statements, which
basically records in \sVar the \emph{variable} of which the value is returned, if any, and then
jumps to the method {exit point}
{; \ie \sVar can take as many values as there are distinct \creturn statements in the  method}.
We associate the location
\(ℓ_\surd\)  to the  {exit} point.
To construct the summary effect, we
define an {invariant} \(φ_{\rname{exit}}\) for this location, that enforces a \emph{lower-bound} on 
every variable in the footprint \(\Footprint m\).
{In \figurename~\ref{fig:method-preambles-n-return}, the predicate $φ_x$ states that, if the method returns a value or reference $x$, \ie \(\sVar = \sSpec x\), its upper bound with \pc should be subsumed by the initial level of the returned value, denoted $\plvl{\retVar\pi}$.}
{When this predicate is propagated backwardly to the initial state along every possible path by the co-reachability algorithm, the upper-bound \(\plvl{\retVar\pi}\) will account for the updates to the returned value.}

Similarly, the heap abstraction \retVar{\hvar{}} is an
upper-approximation of the actual heap abstraction \hvar upon
termination of the method, and we define the invariant \(φ_{\hvar{}}\) to obtain the updates of the method to the heap model.
Depending on whether the method returns a
reference or not,
\(\hvar[π←r]\) denotes {the extension of the  heap abstraction} \hvar with a reference π 
that is an alias of \(r\).
Further, \(\hvar ⊑_{\mathsf H}\retVar{ \hvar}\) denotes a partial order that
holds whenever the abstract value \hvar, restricted to the set of
references \(R' ⊆ R\) pertaining to \(\retVar{\hvar}\), over-approximates a
set of heap configurations,
that is greater  than or equal to the set of heap configurations
represented by \(\retVar{ \hvar}\). This amounts to a pairwise comparison of security levels, \ie one
must have \(\hvar ⊑_{\mathsf H} \retVar{\hvar} ⇒ (∀r∈{R'}, \hlvl[\hvar]r ⊑
\hlvl[\retVar{\hvar}]r)\) (and similarly for any over-approximated fact about
the heap, like a may-alias relation).

{For example, the invariant $φ_{ℓ_\surd}$ defined  for the exit location ${ℓ_\surd}$ in  \figurename~\ref{fig:overview} requires the level of \cv e be subsumed by the initial level of the return value $\pi$. Note that this method has a single return statement, and $\sVar = \sSpec e$ always holds. Additionally, the references  \cv \this and \cv r should not be upgraded upon return, and \cv \this and \cv r remain aliased if they initially aliased each other (\ie \(\hvar ⊑_{\mathsf H} \retVar{\hvar{}}\)).}

\subsubsection{Summary Construction via Co-reachability}
Our approach follows a co-reachability analysis to compute initial states from which the method $m$ does not reach undesired states.  We represent the state space as mappings from locations to predicates on state variables.  We denote $\B_0$ the set of all initially known undesired states, that violate the invariants defined in the semantics (\eg $φ_\rname{{Sink}}$ and $φ_\rname{{Exit}}$ shown in Figure~\ref{fig:method-preambles-n-return}), and should be avoided.
Then, the set of undesired states is back-propagated via a standard co-reachability analysis (embodied by \Coreach{}), that finds all states from which the set of states \(\B_0\) may be reached.
\green{
 This is typically solved using a fixed-point~\citep{Ramadge89,PnueliRosner}.
On a symbolic finite-state system like $S_m$, this computation always terminates, and is traditionally performed using the least fixed-point (\(\mathrm{lfp}\))
{%
\setlength{\abovedisplayskip}{1pt}%
\setlength{\belowdisplayskip}{1pt}%
\begin{equation}
\B_∞ ≝ \mathrm{lfp}\ λ\B_i.\B_0 ∪ \mathrm{pre}(\B_i)\mbox,
\end{equation}%
}%
where \(\mathrm{pre}(\B)\) gives all predecessor states of \B.
}
The resulting mapping
\(\B_∞\) associates each location $\ell$ with a predicate that must \emph{not} hold for every subsequent path from $\ell$ in \(S_m\) to {represent desired executions}.
The predicate \(\B_∞(ℓ_0)\) describes the set of initial states that can reach undesired states \(\B_0\).
Therefore, its complement \(¬\B_∞(ℓ_0)\) is the set of all path conditions that lead to desired executions.

To obtain a guard and an effect for the method, we need to eliminate from \(¬\B_∞(ℓ_0)\) every symbolic variable that does not relate to its calling context.
\green{ In other words, we want to eliminate the part of the security typing environment and heap abstraction that relates to \emph{local} primitive or reference variables.
In the initial states of $S_m$, denoted by $X_0$, \emph{local} primitive and reference variables are set
to low-sensitive, and
\emph{local references are initialized to \Null in the heap abstraction \hvar}.
On the contrary, the security level of every formal argument and \pc is left \emph{unconstrained} in $X_0$.
To perform the aforementioned elimination, we can therefore partially evaluate the complement of  \(\B_∞(ℓ_0)\) {against} the initialization predicate $X_0$.
}
This is done with 
\(\mathrm{cofactor}(¬\B_∞(ℓ_0), X_0)\), which
takes every equality \(x = c\) such that \(X_0 ⇒ (x = c)\), where \(x\) is a variable and \(c\) is a constant, and substitutes \(c\) for \(x\) in \(¬\B_∞(ℓ_0)\).
{This results in a predicate $G$ expressed over the variables of the method support and the method footprint, i.e., $\SumSupport m \cup \SumFootprint m$. We then use a \emph{triangularization} technique~\citep{Hietter-triang} to transform $G$ into an
	summary for $m$.}
\green{
\subsubsection{Triangularization}
The  co-factorization in Algorithm~\ref{alg:synthesis.algorithm} results in a predicate $G$ expressed over the variables of the method support and the method footprint, i.e., $\SumSupport m \cup \SumFootprint m$. We then use a \emph{triangularization} technique~\citep{Hietter-triang} to transform $G$ into an
summary for $m$.
By using this algorithm, we compute the set of assignments \SumEffect m incrementally, by successively: (i) identifying an assignment \(v \assign e\) such that \((v = e) ⇒ G\), for a footprint variable \(v\), (ii) substituting \(v\) by \(e\) in \(G\).
Once this has been performed for every variable of the footprint, \(G\) only involves variables from the support $\SumSupport m$ and constitutes a guard \SumGuard m for the method.
The function $\Triang{G}{V}{E}$, that constructs an summary (\ie pair of guard and effects) from $G$  defined as follows:
\begin{multline*}
	\smaller
\Triang G V E ≝ \\
\begin{cases}
\Triang {\subst G v e} {V\backslash｛v｝} {E ∪ ［v \assign e］}
& v ∈ V\mbox, \\
〈E, G〉 & \text{if~} V = ∅
\end{cases}
\end{multline*}%
where \(e = \Minimize G v\) is a symbolic expression \(e\) that only involves variables in \(G\) and minimizes \(v\) subject to \(G\); \ie one notably has \((v = e) ⇒ G\), {and there does not exist an expression \(e'\) that is not equivalent to \(e\) and such that: (i) \((v = e') ⇒ G\), and (ii) \(e' ⊏ e\) is satisfiable (if \(v\) is a level variable)}.
The search for "minimal" expressions in the triangularization is therefore a key factor in obtaining \emph{precise} effects in summaries.
{For instance, updates that assign ⊤ to level variables of the footprint, albeit safe upper-approximations, would not constitute precise-enough descriptions of the methods' effect.}

Note, however, that there may exist multiple summaries that can be derived from a single predicate \(G\) obtained in Algorithm~\ref{alg:synthesis.algorithm}.
Indeed, the resulting effect is impacted by the order in which each variable from the footprint is selected in \(\Triang{}{}{}\).
Moreover, the domain of footprint variables may have multiple minimal elements, and the selection of a valid symbolic expression \(e\) in \Minimize{}{} may be subject to some arbitrary choice as well.
In our implementation, we define arbitrary but deterministic, non-ambiguous priorities over the full range of variables and their respective domains.
Most notably, the triangularization is parameterized with a sequence of \emph{total	orders} \(（\prec_v）\) for every variable \(v ∈  \SumFootprint m\).
}
\begin{leaveout}
{
\begin{example}[Computing a summary for \Cbar]
\label{exmpl:sum-c-bar-computation}
Assuming the summary of \Cfix given in Example~\ref{expl:summry-4-fix}, the process described above results in the following summary for \Cbar:
	\begin{align*}
		\SumGuard\Cbar = &\ \cv{p0} ⊔ \pc ⊔ \cplvl b ⊔ \cplvl r = ⊥ \\
		\SumEffect\Cbar = &
                \left[\begin{array}{@{}ll@{}}
                      \chlvl{this} \assign \cplvl{this} ⊔ \chlvl{this},&
                      \chlvl r \assign \cplvl{this} ⊔ \chlvl r ⊔ \chlvl{this},\\
                      \cAliasRelProp r {this} \assign \tt,&
                      \plvl{π} \assign \cplvl{this} ⊔ \chlvl r ⊔ \chlvl{this}
                      \end{array}\right].
	\end{align*}
	The guard of this summary indicates that a call to \Cbar satisfies our {objective confidentiality property} if this call happens if the source returns low-sensitive data, in a low context (\(\pc = ⊥\)), the values of its effective arguments \cv r and \cv b are low-sensitive (\(\cplvl b ⊔ \cplvl r = ⊥\)). 
	Furthermore, if \Cbar executes and terminates (\ie \SumGuard\Cbar holds upon call), then: (i) references \cthis and \cv r may share objects (\(\cAliasRelProp r {this} = \tt\)), (ii) the portion of heap pointed to by \cthis may have been upgraded with information from the reference \cthis, (iii) the portion of heap pointed to by \cv r may have been upgraded with information from either (or both) the reference \cthis and objects pointed-to by \cthis, and (iv) the returned value may hold information from \cthis and/or portions of heap reachable via \cthis and/or \cv r.
\end{example}%
}
\end{leaveout}

\begin{figure}
	\begin{ruleset}
		\myrule{Call}{%
			\begin{array}{@{}c@{}}
				\semloc{\stm = \s{[\_ =] r.m(\lits)}} \quad \quad
				T_{\hvar}' = \HeapCappedJoinProj{\SumEffect m}{\nomlvl{⊤}}{\hvar} \qquad \\[1.2ex]
				T_{\hvar} = T_{\hvar}'[π→r'] \text{~if~} \stm=\s{r' = r.m(\lits)}, T_{\hvar}' \text{~otherwise} \\[1.2ex]
				T_{π} = [\plvl x \assignv \SumEffect m(\plvl{π})] \text{~if~} \stm=\s{x = r.m(\lits)}, ∅ \text{~otherwise}
			\end{array}
		}{%
			\begin{array}{@{}r@{~}l@{}}%
				ℓ =
				{\stm };\stms \quad%
				&\trans{\tt, T_{π} \SQMergei u T_{\hvar}}\newloc{\stms}\\
			\end{array}
		}
	\end{ruleset}
	\smaller
	\setlength{\abovedisplayskip}{0pt}%
	\setlength{\belowdisplayskip}{-3pt}%
	where
\begin{align}
	\nomlvl l &\ ≝\ \phantom{\plvl x \assign\,}（\ite{\uVar 
	}{⊥}{l}） ⊔ \pc
	\label{eq:level-cap} \\
	\plvl x \assignv l &\ ≝\ \plvl x \assign （\ite{\uVar 
	}{\plvl x}{l}） ⊔ \pc
	\label{eq:assignv} \\
		φ_{\textsc{Call}} ≝~
& 
\SumGuard m'
\tag{\larger\rname{φ-Call}}\label{eq:invariant4call}
\end{align}
	\caption{Translation rule and objective invariant for method calls
		.}
	\label{fig:method-calls}
\end{figure}

\section{Inter-procedural Security Analysis with  Summaries}
\label{sec:interprocedural}


To obtain a full inter-procedural analysis, we extend our semantics by re\-using summaries
in place of method calls, and detail our technique for modular computation of summaries.
We explain how we handle potential recursion in a subsequent Section.

		\vspace{-0.1in}
\green{\subsection{Call\-graph Construction}
We first compute an over-approximation of the call\-graph for the entire application.
To this end, we need to identify the methods that could be potentially called at any call-site \(x = r.m(\lits)\).

\newcommand{\types}{\ensuremath{\mathbb{T}}\xspace}
\newcommand{\type}{\ensuremath{\mathsf{Typ}}\xspace}
\newcommand{\retype}{\ensuremath{\type_{\mathsf{ret}}}\xspace}
\newcommand{\basetype}{\ensuremath{\type_{\mathsf{recv}}}\xspace}
\newcommand{\argtype}[1]{{\ensuremath{{\type_{a,#1}}}\xspace}}
\newcommand{\typesRel}{\ensuremath{\prec}\xspace}

Let the pair $\langle \types, \typesRel \rangle$ represent a \emph{type hierarchy} for the whole application where $\types$ is the set of types, i.e., the set of classes, interfaces and primitive types, and $\typesRel$ is a partial ordering on $\types$ showing the inheritance relations among the object types. For the types $t, t' \in \types$, the relation $t \typesRel t'$ holds if (i) $t$ and $t'$ are classes and $t$ inherits from $t'$ by extending it, or (ii) the class $t$ implements the interface $t'$.
We further represent a method signature by $V = \langle 
\basetype, \mathit{name}, \argtype 1, \ldots, \argtype n \rangle$ where 
$\basetype$ is the type of the receiver reference, $\mathit{name}$ is the method name, and $\argtype i$ represents the type of the $i$-th argument.

A method $m$ with signature $\langle 
\basetype, m, \argtype 1, \ldots, \argtype n \rangle$ is a potential target for a method call $x = r.m(w)$ if and only if the following conditions hold:
\begin{itemize}
    \item $\type(r) \typesRel \basetype$,
    \item $|w| = n$, and
    \item $\type(w_i) \typesRel \argtype i$ for $0 \leq i \leq n$
\end{itemize}
For static calls (\ie $r$ is a class name), $\type(r) = r$. 
We analyze the entire application and build an over-approximated call\-graph whose edges show potential target methods for a method call.
}

\subsection{Re\-using Method Summaries}
\label{sec:re-using-summaries}
{We first compute an over-approximation of the call\-graph for the entire application.}
%

We combine the set of reusable summaries already computed for a call-site to {encode} method calls.
From the call\-graph, we can obtain (an over-approximation of) the set of every method that may be targeted at any call-site \(r.m(\lits)\).
We can therefore obtain the corresponding summaries
.

Let \TargetSummaries m be the set of summaries of all methods that can be dispatched for a method call $r.m(w)$.
We first take this set of summaries, and substitute the formal arguments for the corresponding effective parameters \(r\) and \lits; this produces a set of summaries that all have identical supports and footprints.
{Then, to obtain a single summary to use in place of \(r.m(\lits)\), we combine the summaries' guards and effects. The combined guard should avoid insecure states for any possible target method, and the combined effects should include the effects of all possible target methods. Therefore, we define the combined guard as the conjunction of the summary guards of all possible target methods, and the combined effect as an over-approximated union of their effects.}
\begin{leaveout}
The symbolic structure of summaries makes this operation straightforward, primarily since, thanks to the above renaming, the summaries to combine have identical supports and footprints. Thus, joining a set of summaries boils down to:
\begin{enumerate*}[(i)]
	\item construct the \emph{conjunction} of all the guards: this is required since our analysis domain does not capture object types and is therefore unable to handle virtual dispatch precisely;
	\item join their effects to produce an over-approximation of effects using the \emph{merge operation} \SQMergei{}.
\end{enumerate*}
\end{leaveout}
Therefore, the combined summary to be used in place of $r.m(w)$ is defined as
{
  \small
  \begin{equation*}
    \langle\SumGuard m, \SumEffect m\rangle ≝
    \langle%
    \underset{\langle\SumGuard{}, \_\rangle ∈ \mathrlap{\TargetSummaries m}}{⋀} \SumGuard{},
    \underset{\langle\_, \SumEffect{}\rangle ∈ \mathrlap{\TargetSummaries m}}{\SQMerge u}\SumEffect{}%
    \rangle%
    \mbox.
  \end{equation*}
}%
To account for %
the implicit branching behavior that is induced by virtual method
dispatch on the actual class of the target object
, we further 
upgrade the context level \pc by substituting it in the guard and
effect as follows:
{
	\begin{equation*}
		\SumGuard m' ≝ \begin{cases}
			\subst{\SumGuard m\!}{\pc}{\pc ⊔ {\hlvl[]r}} & \mbox{if~} |\TargetSummaries m| > 1 \\
			\SumGuard m & \mbox{otherwise}
		\end{cases}
\end{equation*}}%
where {$\hlvl[]r$ is the security level of the receiver object (pointed-to by \(r\))} and \subst e v l denotes the substitution of security level
expression \(l\) for variable \(v\) in \(e\).
The substitution in effects is performed in every expression 
on the right-hand side of assignments.

{
In order to capture implicit flows, we consider two execution modes determined by a state variable \uVar: (i) in nominal mode (\uVar = \ff), updates to security levels reflect explicit information flows, and (ii) in upgrade analysis mode (\uVar = \tt), the information flow from the high execution context pc to all the relevant branches are captured.
}
The rule \rname{Call} given in
\figurename~\ref{fig:method-calls} encodes the semantics of method calls
in a semantic location ℓ, where the invariant \ref{eq:invariant4call}  ensures that there is no illegal flow by the method $m$, \(T_{π}\) updates the security type of any assigned returned value or
reference $x$ to the security type of return value $\pi$ in {\SumEffect m},
and \(T_{\hvar}\) applies the method effect on the heap model of the caller.
{We apply the function \HeapCappedJoinProj{\SumEffect m}{\nomlvl{⊤}}{\hvar} to incorporate the implicit flows to the heap,
}
where \HeapCappedJoinProj{\mathit{eff}}
l \hvar constructs a new update function that applies the update function \(\mathit{eff}\)
on the heap model \hvar
, in such a way that every security level upgrade by
\(\mathit{eff}\) is capped with  the security level expression \(l\).
The correctness of this construction is due to the fact that
 the update encoded by 
 effects to the security level of a
 potentially mutated object is always lower or equal than the context
 level \pc ---unless the guard does not hold in high-context, in which
  case the method call is made unreachable via the enforced invariant.
 In other words,
\pc is an upper bound on the level assigned by
\(m\)'s effect to
the security level associated with
any portion of heap that is potentially mutated by \(m\).
Hence, since \(\nomlvl{⊤} = \pc\) in upgrade analysis mode---\cf
Eq.~(\ref{eq:level-cap})---, every component of the heap typing
environment that corresponds to a potentially mutated portion of heap
is weakly upgraded with the context level.
In nominal mode, however, the assignments to the typing environment
are performed directly since \(\nomlvl{⊤} = ⊤\).

\green{
Note that the heap abstract domain that constitutes the footprint of
the applied effect corresponds to a subset of all reference variables
in the caller method \(m'\).
Therefore, the capped application operation 
needs to pessimistically propagate every level upgrade and alteration
of heap-related relation to references 
that may point
to portions of heap that are (indirectly) reachable by reference
arguments in \lits.
%
The remaining transformation performed by the rule \rname{Call} to
obtain \(T_{\hvar}\) essentially consists in 
renaming variables to substitute the return symbol π for any reference
variable \(r\) assigned upon \(m\)'s termination.
\nbrem[yep]{At last, if \(m\) returns a reference that is assigned to
	a variable \(r\), then \(T_{\hvar}\) is obtained as
	\(T_{\hvar}'[π→r]\), that:
	\begin{enumerate*}[(i)]
		\item discards any assignment to a symbol variable that is made of
		reference \(r\) (\eg any assignment \(\hlvl[\hvar] r \assign e\),
		or \(\aliasRel[\hvar] r s \assign e\));
		\item substitutes any return symbol π for \(r\) in the remaining
		assigned variables (\eg an assignment \(\aliasRel[\hvar]{π}s \assign
		e\) becomes \(\aliasRel[\hvar] r s \assign e\)).
\end{enumerate*}}
}
\begin{leaveout}
	\begingroup%
\begin{example}[Reusing \MethSum \Cfix]
	\label{exmpl:reusing-sum-c-fix}%
	We have represented in \figurename~\ref{fig:c-bar-example-scfg} the
	update function \(T_\fix\) and invariant \(φ_{ℓ_4}\) that correspond
	to the method call statement in location \(ℓ_4\).
	Assuming that the target of this method call consists in a single
	implemented method whose summary is the one given in
	Eq.~(\ref{eq:c-fix-example-summary}) for \Cfix, \ie
	\(\TargetSummaries{\stm_4} = \big\{\MethSum\Cfix\big\}\), then one
	obtains \(φ_{ℓ_4} = 
	\SumGuard\Cfix\) and \(
	T_{\fix} = \HeapCappedJoinProj{\SumEffect
		\Cfix(\hfvar\fix)}{\nomlvl{⊤}}{\hvar}\).
	In this particular instance, \SumEffect\Cfix does not require any
	renaming, and \(T_{\fix}\) is actually derived via the capped
	application as:

	\footnotesize{
	\[
	T_\fix = ［%
	\begin{array}{@{}r@{~\assign~}l@{}}
		\hlvl\this & \hlvl\this\,⊔\,\ite{\!\plvl\this ⊑ \nomlvl{⊤}}{\plvl\this}{\nomlvl{⊤}} \\
		\hlvl r & \hlvl r\,⊔\,\ite{\aliasRel r \this}{\ite{\!\plvl\this ⊑ \nomlvl{⊤}}{\!\plvl\this}{\nomlvl{⊤}}}{⊥}
	\end{array}%
	］\mbox.
	\]
}
\end{example}
\endgroup%
\end{leaveout}

\subsection{Full Summary-based Program Analysis}
\label{sec:full-bottom-up}

Let us denote \Summary m D the procedure given above for computing the
summary of an implemented method \(m\) based on a set \(D\) of
summaries computed for every method it may call.
Then, computing a correct summary for every implemented method in a
program boils down to 
searching a
solution 
to the following system of equations:
\begin{equation}
\label{eq:full-fixpoint}
\begin{cases}
\begin{aligned}
	\MethSum {m_1} &= \Summary{m_1}{\{\MethSum m~|~{m ∈ \mathsf{Deps}（m_1）} \}} \\[-1.2em]
	&\shortvdotswithin{=}\\[-1.6em]
	\MethSum {m_n} &= \Summary{m_n}{\{\MethSum m~|~{m ∈ \mathsf{Deps}（m_n）} \}}
\end{aligned}
\end{cases}
\end{equation}
where \MethDeps m gives the set of methods that is potentially called
by \(m\).
Each unknown variable in Eq.~(\ref{eq:full-fixpoint}) is an summary for
a 
method of the program, and 
this system naturally induces a dependency graph \(G\) where nodes are
methods; \(G\) corresponds to an over-approximation of the call-graph.
If the program under scrutiny does not feature (potentially) recursive
methods (\ie \(G\) is a\-cyclic), then one can use a simple
topological traversal of \(G\) to compute everysummary.
Yet, this solution is not applicable if \(G\) is cyclic, \ie involves
recursive methods.
To address this issue, we employ a technique that is well-established
in the abstract interpretation community, that consists in following a
\emph{chaotic iteration
strategy}~\citep{Bourdoncle:1993:ChaoticIterationStrategiesWithWidenings}
to find a fixed-point solution for Eq.~(\ref{eq:full-fixpoint}).

This solution is sound if \Summary{m}{⋅} is monotonic for any fixed
\(m\). 
 The summary that we compute for each method $m$ belongs to a
 domain that can additionally be equipped with a least element and a
 partial order 
 to form a \emph{lattice of summaries} \(\langle\Summaries m, ⊑_{\Summaries m}
 \rangle\).
We informally posit that the {possible  summaries for $m$} are ordered according
 to decreasing levels of permissivity, so that the least element
 corresponds to the most permissive summary.
The bottom element of this lattice corresponds to a
summary 
with a guard that
always holds (it is \tt), and an effect that 
returns a low-level value, if any, and does not perform any
object mutation.
We prove that the procedure \Summary{m}{.} is monotone \wrt the
 lattice  \(\langle\Summaries m, ⊑_{\Summaries m} \rangle\).
 Note that relying on a chaotic iteration strategy may also require the
definition of a suitable \emph{widening operator} that merges
summaries in a way that accelerates stabilization to the fixed-point
(or force convergence if \Summaries m is infinite
)
.
\nbrem[yep]{If summaries are built out of symbols defined on finite
domains, then the widening operator can actually be
constructed 
by lifting widening operators for guards and assigned expressions in
effects, such as the ones defined
by~\citet{Mauborgne1998AbstrInerprWithBDDs} for 
BDDs.}
Note, though, that our current implementation does not involve the
above widening step to accelerate convergence, and relies on the
finiteness of the lattice \Summaries m instead.
 We will see in Section~\ref{sec:experimental-results} that this choice
 is comforted by our empirical experiments.
\nbrem[yep]{Indeed, our empirical experiments suggest that, with our current
choice of symbolic heap abstract domain, for a small number of methods
only, very few additional applications of the \Summary{}{} procedure
are required to find a solution.}

 \begingroup%
{
In our running example, assume that \flushBits may call \cv{FileImageOutputStream}.\Write. Then, the
 process for computing a summary for \Write roughly 
 consists in assuming a most permissive summary, say, \(\MethSum\flushBits^0\) for \flushBits,
then compute an summary for \Write by using \(\MethSum\flushBits^0\), and then
compute summaries for each method until stabilization, one after the
other:
  \begin{enumerate}[(1),nosep,leftmargin=*,labelindent=0pt]
  \item generate a most permissive summary \(\MethSum\flushBits^0\) for \flushBits;
  \item compute \(\MethSum\Write^0 = \Summary\Write{｛\MethSum\flushBits^0,\MethSum{{\cv{RandomAccessFile}.\Write}}^0｝}\);
  \item compute \(\MethSum\flushBits^{i+1} = \Summary\flushBits{｛\MethSum\Write^i｝}\) until stabilization, \ie \(\MethSum\flushBits^{i+1} = \MethSum\flushBits^i\);
  \item compute \(\MethSum\Write^{i+1} = \Summary\Write{｛\MethSum\flushBits^{i+1},\MethSum{{\cv{RandomAccessFile}.\Write}}^{i+1}｝}\) and then repeat from step (3) until stabilization, \ie \(\MethSum\Write^{i+1} = \MethSum\Write^i\).
  \end{enumerate}
}
 \endgroup%

{To prove that our approach for summary computation is sound, we need to prove that our procedure \Summary {m}{.} is monotonic  \wrt \( \langle\Summaries m, ⊑_{\Summaries m} \rangle \).
Let \(\MethSum{} = \langle\SumGuard {}, \SumEffect {}\rangle\) and  \(\MethSum{}' = \langle\SumGuard {}', \SumEffect {}'\rangle\) be two summaries.
We state \(\MethSum{} \not\sqsubseteq_{\Summaries m} \MethSum{}'\) if and only if $\SumGuard {}'$ is more permissive than \SumGuard {}, \ie \(\SumGuard {}' (s) \implies \SumGuard {}(s) \) for all states $s$.

}
\begin{property}
	Given a method \(m\), \Summary{m}{⋅} is monotonic \wrt \( \langle\Summaries m, ⊑_{\Summaries m} \rangle \).
\end{property}
\begin{proofSketch}
	{See  Appendix~\ref{sec::summary.computation.proof}.}
\end{proofSketch}


\section{Formal Correctness}
We prove that any method guarded with its IF-summary guard guarantees termination-insensitive non-interference~\citep{sabelfeld2003language}.
This notion states that, for any initial  states $s$ and $s'$ whose secret parts may only differ, the observations sequence of the program running from the states $s$ and $s'$ will either be the same,
or one is a prefix of the other.
The reason for the latter case is that this notion is a termination-insensitive property.
In addition to the symbolic variables that belong to the semantics for each called method (\ie the heap model and the security typing environment), we also need to consider an extended model of the \emph{heap} used in the execution, as well as the \emph{call\-stack} with values for local variables.
Let \(\S \) denote the (symbolic) full semantics of a program,
and $\FSM\S=\FSMdef$ be an automaton that describes its extended
semantics,
where \smStates is the set of states, $\mathcal{I}$ is the set of inputs,
$\smTrans \subseteq \smStates \times \smStates$ is the set of transitions and $\smStates_0$ is the set of initial states.
A program
state \(\qstate \in \smStates\) is defined as \( \methFrame 
:: \mStack \) where \(\methFrame\) is the  frame of the currently executing
method, and \mStack is the stack of calling contexts.
A method frame
is a tuple \(\mframeX  {ℓ} \pval    {\mHeapVal} {\X} \) where
ℓ
is the current location,
\pval   is the
valuation the
method's primitive variables $\PrimVars$,
\mHeapVal is the valuation of the extended heap's variables \havar,
and \(\X
\)
is the current valuation of the state variables in the semantics. \X consists of  (i)
{\(\VarsTypes: \PrimVars \cup \red{R}
	→ \LL\)} to show the \emph{security typing environment}
for {the primitive variables $\PrimVars$ and references $R$},
and (ii) $\HeapTypes: {R} → \LL$  to denote objects security typing environment.
We then define non-interference based on a \emph{low-equivalence relation} between two states $\qstate_1$ and $\qstate_2$, that states the public parts of the states  are indistinguishable.

	We say  two valuations $\pval  _1 $ and $\pval  _2 $ of primitive variables are \emph{indistinguishable},
	denoted by \(\pval  _1 =_{\low} \pval  _2\),
	iff $\pval  _1(v) = \pval  _2(v)$ for all $v ∈ \PrimVars$ where \(\plvl v=\low\).
We extend this notion to portions of heaps reachable by sets of references.
%
We say two extended heaps are indistinguishable, if heap-related relations and  primitive fields of their low-sensitive references are identical, \ie (i) the reference graphs corresponding to their low-sensitive portions of the heaps are \emph{isomorphic}, and (ii) the valuation of primitive fields of their low-sensitive references are~identical.
A reference graph shows the heap relations among references.

We use the notation $\qstate(a)$ to show the value of $a$ in the state $\qstate$.
Furthermore, we say   {$\qstate_1(\HeapTypes) \sqsubseteq \qstate_{2}(\HeapTypes)$}, if and only if for all $r \in R$, $\qstate_1(\hlvl[] r) \sqsubseteq \qstate_2(\hlvl[] r)$.
We say two states \(\qstate_i= \mframeX{ℓ_i}{\pval  _i}{\mHeapVal_i}{\X_i}::\eta_i \)
	\(i \in \{1,2\}\)
	are \emph{compatible}, denoted by $\qstate_1 \approx \qstate_2$, iff
	\(\qstate_1(ℓ)=\qstate_2(ℓ)\),
	\(\qstate_1(\pval)=_{\low}\qstate_2(\pval)\),
	\({\qstate_1(\mHeapVal) =_{\low} \qstate_2(\mHeapVal)}\),
	\(\qstate_1({\VarsTypes})=\qstate_2(\VarsTypes)\),
	\(\qstate_1({\HeapTypes})=\qstate_2(\HeapTypes)\), and
	{\( \qstate_1( \mStack) =  \qstate_2( \mStack) \)}.
	Let $\qstate \xrightarrow{\beta}_* \qstate'$ be an execution of full semantics  with a non-zero length (\ie the reflexive and transitive closure of the concrete transition relation $\smTrans$) from the state $\qstate$ to the state $\qstate'$, where $\beta \in \{o,\bot\}$.
	This execution either ends by executing a statement that outputs on a channel (%
	\ie $\beta=o$)
	or {continues its execution without making an observation} (\ie $\beta = \bot$).
	We denote an execution that never reaches an observation point by $\qstate \xrightarrow{\bot}_* $.
	We  define the low-equivalence relation between two states of the full semantics~
	as:
	\begin{definition}[{Low-Bisimulation}]\label{def::low-bisimulation-relation}
		Two  states \(\qstate_1\) and \(\qstate_2\) are low-bisimilar, denoted by $\qstate_1 \lowbisim \qstate_2$, iff \(\qstate_1 \approx \qstate_2\), and  if $\qstate_1\xrightarrow{o}_* \qstate'_1$,
		then either (a) there exists $\qstate'_2$ such that $\qstate_2 \xrightarrow{o}_* \qstate'_2$
		and $\qstate_1' \lowbisim \qstate'_2$,
		or (b) $\qstate_2 \xrightarrow{\bot}_* $,
		and vice versa.
	\end{definition}

	If for any two arbitrary initial states  $\qstate_0$ and $\qstate_0'$ of a program $c$ where $\qstate_0 \approx \qstate'_0$, it holds that $\qstate_0 \lowbisim\qstate'_0$, we say  $c$ guarantees noninterference.

	\begin{theorem}[Noninterference]\label{thm::Non.Interference}
		For any method $m$ with a security summary \( \langle \SumGuard m, \SumEffect m \rangle \), and any initial states \( \qstate_1\) and  \( \qstate_2\) where
		\(\qstate_1\approx \qstate_2\) and
		$\qstate_i\models \SumGuard m$,
		$i \in \{1,2\}$, it holds that $\qstate_1 \lowbisim\qstate_2$.
	\end{theorem}

 \begin{proof}
See Appendix~\ref{sec::noninterference-proof}.
 \end{proof}
	%

	\newcommand\NumMeth{\#meth\xspace}%
\newcommand\NumVars{\(|\!V\!|\)\xspace}%
\newcommand\FootSize{\(|\!F\!|\)\xspace}%
\newcommand\MethLen{\(|\!m\!|\)\xspace}%
\newcommand\AToM{\(t\!/\!m\)\xspace}%
\newcommand\Processed{Processed\xspace}%
\newcommand\NumSkipped{\#Skipped\xspace}%
\newcommand\NumGuarded{\#guarded\xspace}%
\newcommand\NumUpdates{\#flows\xspace}%
\newcommand\NumSecure{\#fully secure\xspace}%
\newcommand\ClockTime{\(T_{\mathsf{wc}}\)\xspace}%
\newcommand\TimeMethodAvg{$T_\mathsf{meth}$\xspace}%
\newcommand\TotTime{\(T_\mathsf{tot}\)\xspace}%

\newcommand\TaintExpr{{\sffamily TA}\xspace}%
\newcommand\ExplicitConf{{\sffamily EFA}\xspace}%
\newcommand\ImplicitConf{{\sffamily IFA}\xspace}%

\begin{figure}
	\centering%
	\colorlet{textfg}{black}%
\begingroup
\sffamily \smaller
\color{textfg}%
\tikzexternaldisable%
\begin{tikzpicture}[>=stealth', node distance = 4mm]
  \tikzset{every node/.style={rectangle, draw = textfg,fill = white, align = center, inner sep = 1.4pt}};
  \node [rounded corners, minimum height=7em] (syrs) {\Symmaries\\(/\ReaX)};

  \coordinate [below = 2.5em of syrs.west] (x);
  \node [left=of x] (secstubs) {\strut{.secstubs}};
  \path [->
  ] (secstubs) edge (x);

  \coordinate (x) at (syrs.west);
  \node [left=of x,xshift=-2pt,yshift=2pt] {\strut{.meth}};
  \node [left=of x] (meth) {\strut{.meth}};
  \node [left=of x,xshift=2pt,yshift=-2pt] {\strut{.meth}};
  \path [->,shorten <= 2pt] (meth) edge (x);

  \coordinate [above = 2.5em of syrs.west] (x);
  \node [left=of x] (classes) {\strut{.classes}};
  \path [->] (classes) edge (x);

  \coordinate (x) at ([yshift = 1.4em]syrs.east);
  \node [right=of x,xshift=-2pt,yshift=2pt] {\strut{.secsum}};
  \node [right=of x] (secsum) {\strut{.secsum}};
  \node [right=of x,xshift=2pt,yshift=-2pt] {\strut{.secsum}};
  \path [->] (x) edge (secsum);

  \coordinate (x) at ([yshift =-1.4em]syrs.east);
  \node [right=of x] (meth-stats) {\strut{.meth\_stats}};
  \path [->] (x) edge (meth-stats);

  \node [fit=(classes)(meth), draw = none, fill = none] (classes-meth) {};


  \node [left=of classes-meth, rounded corners,minimum height=4em, align=center] (jsym) {JSymb\\(/soot)};
  \path [->] (jsym.east |- classes.west) edge (classes.west);
  \path [->] (jsym.east |- meth.west) edge (meth.west);

  \coordinate (x) at (jsym.west);
  \node [left=of x,dotted,align=left] (java-inputs) {%
    \strut{.java}, \strut{.class},\\
    \strut{.jar}, \strut{.apk}…%
  };
  \path [->] (java-inputs) edge (x);
\end{tikzpicture}%
\hrule width 0pt              
\endgroup%
 %
	\caption{\Symmaries' workflow.}%
	\label{fig:symmaries-basic-workflow}%
\end{figure}%
\section{Implementation and Evaluation}
\label{sec:experimental-results}
\green{To support the proposed approach,}
we have implemented a new toolchain whose architecture is shown in \figurename~\ref{fig:symmaries-basic-workflow}. The \textsf{JSymb} module takes an application and uses \soot~\citep{vallee2010soot} to generate its intermediate Java bytecode representation in \textsf{Jimple}, which is then translated into the input format of \Symmaries. The input includes a definition of data types (i.e., classes, fields, and inheritance relations, provided in a \textsf{.classes} file), a list of methods to analyze, and optional specification of behavior for excluded methods (e.g., native code or any other library or application method, in a \textsf{.secstubs} file).
\textsf{JSymb} is implemented in \Java with about 5K lines of code (LOC), while \Symmaries is written in \textsf{OCaml} and comprises about 18K LOC.

\green{To accommodate   potentially computationally intensive analysis, we have equipped \Symmaries with customizable timeout and memory-limitation configurations that interrupt any long-running or
memory-greedy analysis; methods skipped in this way are then treated as third-party%
		{, which
			means that they are assumed to satisfy the desired IFC property, but may only be called in low-context and accept low-sensitive outputs. }
		Note, however, that a method being skipped does not imply that it cannot be analyzed at all; rather, it often indicates that we need to  perform the analysis under a different configuration.
	}

Our evaluation is  guided by the following research questions:
	\begin{enumerate}[leftmargin=*,labelindent=0pt]
		\item \textbf{RQ1:} How precise is \Symmaries in synthesizing security summaries for Java programs?
		\item \textbf{RQ2:} How scalable is \Symmaries in handling real-life applications?
		\item \textbf{RQ3:} How usable is \Symmaries?
		\item \textbf{RQ4:}  How does it compare with similar tools?
	\end{enumerate}
	\newcommand{\LPrecise}{{\sffamily LPrec}\xspace}%
	\newcommand{\HPrecise}{{\sffamily HPrec}\xspace}%
	\vspace{-0.1in}

	\subsection*{RQ1: Precision}
	We applied \Symmaries to the Java API library (OpenJDK 8) to extract its security specifications.
	The Java API library serves as a suitable benchmark for our analysis, as it is one of the most widely used and well-established standard libraries. Almost every Java program relies on its core packages, making it a relevant target for extracting security specifications.
	In our experiments, we focused on two main packages, \texttt{java.io} and \texttt{java.util}, which cover I/O communication and data collections. We extracted all their dependencies, including those from other packages, to ensure completeness of the analysis.
    However, we excluded methods involving concurrency and reflection, as our tool does not support these features as well as native code. For all excluded methods, we provided security specifications semi-manually by consulting the Java API documentation and source code to capture their behavior.
    We also employed AI to generate  an initial draft specification in our annotation language, and then manually verified and corrected these specifications, as they contained several mistakes.
	To define the set of sinks and sources, we extracted methods that directly write to or read from external sources, either via native methods or through memory writes. In our approach, sources are symbolic, \ie their security labels are unknown and may be either high-sensitive or low-sensitive depending on the context. For example, the method \texttt{java.io.FileInputStream.readByte}, which is a native method that reads from a file input stream, is a symbolic source that can act as either low or high depending on the sensitivity of the file being read.
    Therefore, the method \cv{read}, which calls \cv{readByte}, also includes flows from \cv{p0} to \cv{r1} in its summary (\cv{objlevel\_r1}), where \cv{p0} represents the security level of the data read by \cv{readByte}. In other words, if \cv{readByte} returns high-sensitive data (i.e., \cv{p0} is high-sensitive), then \cv{objlevel\_r1} will also become high-sensitive:

    \begin{minipage}{\linewidth}
\begin{lstlisting}[mathescape=true,language=meth,basicstyle=\footnotesize\ttfamily]
int java.io.FileInputStream:read (byte[] r1) {
  // --- Guard ---
  guard := tt;
  // --- Updates to security levels ---
  level_ret := pc ⊔ level_r1 ⊔ level_this ⊔ objlevel_r1 ⊔ objlevel_this;
  objlevel_r1 := p0 ⊔ pc ⊔ objlevel_r1;
  ....
}
\end{lstlisting}
    \end{minipage}\hfill


	We evaluated the \emph{precision} of \Symmaries, defined as the proportion of methods with correct specifications. Precision can be assessed along three aspects: the precision of security guards, the correctness of extracted information flows for each method, and the precision of inferred aliasing relations. However, due to the large number of methods requiring manual verification, we measured precision based on the precision of the security guards only. We argue that this provides a reasonable indication of the overall tool precision, as the main source of imprecision arises from over-approximation introduced by the co-reachability algorithm, which is applied once to the full program semantics. The resulting output is then decomposed into the three respective components. Hence, any over-approximation caused by the complexity of the analysis often affects all three dimensions, and we believe that studying one aspect can provide a reasonably reliable estimate of the overall precision.

    \newcommand{\skipped}{{\#skipd}\xspace}
\newcommand{\strict}{{\#strict}\xspace}
\newcommand{\aliasing}{{\#aliasing}\xspace}
\newcommand{\precision}[1]{prec#1\xspace}
\begin{table*}[ht]
\smaller\smaller
\centering
\caption{Precision Results.}
\begin{tabularx}{\textwidth}{l |
    >{\centering\arraybackslash}X |
    *{5}{>{\centering\arraybackslash}X} |
    *{4}{>{\centering\arraybackslash}X}}
\toprule
 &  & \multicolumn{5}{c}{{\HPrecise} results} & \multicolumn{4}{c}{{\LPrecise} results} \\
\cmidrule(lr){3-7} \cmidrule(lr){8-11}
\textbf{Package} & \NumMeth& \skipped & \strict &\aliasing & \precision{1}(\%) & \precision{2}(\%)
                 & \skipped & \strict & \precision{1}(\%) & \precision{2(\%)} \\
\midrule
\cv{java.util}       & 1460 & 201 & 251 & 523 & 82.8  & 99.8 & 1 & 31 & 98  & 100 \\\hline
\cv{java.io}         & 546  & 27  & 77  & 208 &85.9  & 99.8 & 0 & 5  & 99.1  & 99.4  \\\hline
\cv{java.lang}       & 278  & 9   & 39  & 88& 86  & 100& 0 & 0  & 100 & 100 \\\hline
\cv{java.nio}        & 140  & 12  & 37  & 65 & 73.6  & 89.4 & 0 & 0  & 100 & 100 \\\hline
\cv{java.text}       & 42   & 2   & 4   & 20&  90.5  & 100& 0 & 0  & 100 & --    \\\hline
\cv{java.security}   & 24   & 6   & 3   & 11 & 87.5  & 100& 0 & 1  & 95.8  & 100 \\\hline
\cv{javax.security}  & 12   & 1   & 3   & 2 & 75  & 50 & 0 & 0  & 100 & 100 \\\hline
\midrule
\textbf{Total}  & 2502 & 258 & 414 & 919 & 83.5  & 99 & 1 & 37 & 98.5  & 99.9  \\
\bottomrule
\end{tabularx}
\label{table::precision.results}
\end{table*}

   The precision of \Symmaries depends on the heap model used in the analysis. Hence, we conducted experiments with two different heap models: a less precise in modeling aliasing relationships \LPrecise, and another more precise \HPrecise, being field-sensitive. Note that our analysis is still flow-sensitive, object-sensitive and context-sensitive for both heap models.
    Table~\ref{table::precision.results} shows the precision results of the two heap models,
    \HPrecise and \LPrecise, for the analyzed Java packages. The column \NumMeth reports the total  number of methods checked in each package and its subpackages, \skipped counts the methods that the tool failed to analyze due to memory or time out, \strict indicates the number of methods with strict guards, \aliasing shows the number of methods that update aliasing relations, and \precision{1} represents the percentage of methods with precise guards. We checked the precision of guards semi-automatically. We expect that all methods that do not reach a sink method should have a \texttt{true} guard. If a method does not reach a sink but has a non-\texttt{true} guard, we classify it as strict. For methods that reach a sink, we checked the strictness of the guards manually.

\HPrecise achieves a precision of 83.5\%, while \LPrecise reaches 98.5\%. Although one might expect \HPrecise to yield more precise results, the larger number of variables it uses to model the heap
leads to more complex analysis and makes it more susceptible to over-approximation during guard
propagation. Moreover, the precision of this heap model has a greater impact on the aliasing
specification. To fully compare the two domains, one would need to examine their inferred aliasing
specifications as well. However, our goal here is not to compare these heap models directly, but
rather to illustrate the precision of the overall algorithm under different heap assumptions. Moreover, we observed that methods which reach a skipped method often seem to have strict guards, since the tool treats such methods as third-party code with pessimistic behavior. This imprecision is propagated to all methods that directly or indirectly call skipped ones. In other words, our precision results are affected by scalability limitations, which reduce overall precision. By removing the skipped methods along with all callers that depend on them, the precision improves as shown in \precision{2}.  We believe our precision results are very promising. Even if the tool fails to analyze certain methods, a user can manually provide a summary for the method or exclude it, while still being able to analyze the rest of the application.


	\newcommand\bookkeeperapp{\textsf{Bookkeeper}\xspace}%
	\newcommand\bookkeepername{\texttt{Apache BookKeeper}\xspace}%
	\newcommand\fortressapp{\textsf{Fortress}\xspace}%
	\newcommand\fortressname{\texttt{Apache Fortress}\xspace}%
	\newcommand\rocketapp{\textsf{RocketMQ}\xspace}%
	\newcommand\rocketname{\texttt{Apache RocketMQ}\xspace}%
	\newcommand\kerbyapp{\textsf{Kerby}\xspace}%
	\newcommand\kerbyname{\texttt{Apache Kerby}\xspace}%
	\newcommand\secretManagerapp{\textsf{SecretManager}\xspace}%
	\newcommand\secretManagername{\texttt{AWS Secrets Manager}\xspace}%
	\newcommand\guavaapp{\texttt{Guava}\xspace}%
	\newcommand\guavaname{\textsf{Google Guava}\xspace}%
	\newcommand\keycloackapp{\textsf{Keyloack}\xspace}%
	\newcommand\keyloackname{\texttt{Keycloack}\xspace}%
	\newcommand\eventBusname{\textsf{EventBus}\xspace}%
	\newcommand\snappyapp{\textsf{{Snappy-java}}\xspace}%
	\newcommand\snappyname{\texttt{{Snappy Java}}\xspace}%
	\newcommand\zxingapp{\textsf{{ZXing}}\xspace}%
	\newcommand\HSQLDBname{\texttt{{HSQLDB}}\xspace}%

	{
		\begin{table}
			\centering
			\smaller\smaller
			\begin{threeparttable}
				\caption[]{Open-Source Projects List}
				\label{tab::apps.list}
				\begin{tabular}
					{@{}p{2.3cm}|p{11.3cm}@{}}
					\toprule
					Name & Description \\
					\midrule
					\href{https://bookkeeper.apache.org}{\bookkeepername} &
					BookKeeper is a scalable storage service designed for real-time workloads that stores record entries in sequences called ledgers and replicates them across multiple servers for reliability. \\
					\hline
					\href{https://rocketmq.apache.org}{\rocketname} &
					A distributed messaging and streaming platform. \\
					\hline
					\href{https://directory.apache.org/kerby/}{\kerbyname} &
					A common Java-based Kerberos implementation of the Kerberos V5 protocol. \\
					\hline
					\href{https://github.com/aws/aws-secretsmanager-caching-java}{\secretManagername} & An service to manage sensitive information such as credentials, API keys, etc. \\
					\hline
					\href{https://greenrobot.org/eventbus/}{\eventBusname} &
					An open-source library for event-based publish/subscribe communication in Android and Java applications. \\
					\hline
					\href{https://www.keycloak.org/}{\keycloackapp} &
					An open-source project for identity and access management to create a user database with custom roles and groups. \\
					\hline
					\href{https://xerial.org/snappy-java/}{\snappyname} &
					A fast compressor/decompressor. \\
					\hline
					\href{https://github.com/zxing/zxing/tree/master}{\zxingapp} &
					ZXing ("zebra crossing") is an open-source, multi-format barcode image processing library implemented in Java.
					\\\hline
					\bottomrule
				\end{tabular}
			\end{threeparttable}
			\vspace*{-1ex}
		\end{table}
	}
	\vspace{-0.1in}

	\subsection*{RQ2: Scalability Evaluation}
	This part of our evaluation investigates the scalability of our approach.
	We collected several popular \emph{real-life open-source projects} from GitHub and analyzed them to infer their summary specifications. The list of selected projects, along with their descriptions, is provided in Table~\ref{tab::apps.list}, ranging from 2K to approximately 220K lines of Java bytecode.
	All experiments were conducted on a desktop computer with 32~GB of memory and a 3.5~GHz Quad-Core Intel Core i7 CPU running Ubuntu Linux.
	To obtain reasonable timing results, we enforced a timeout of 5 minutes and restricted memory usage to a maximum of 8~GB.
	Any method summarization that exceeded either of these limits was skipped and treated as a third-party method with a pessimistic specification.

	\begin{table}
		\smaller  
			\caption{Analysis Results for the Selected Open-Source Projects}
			\label{tab//web.app.experiments}
\begin{tabular}
 {p{27mm}|p{1cm} | p{0.5cm}|p{10mm}|p{8mm}| p{0.8cm} |p{0.8cm}| p{10mm} | p{10mm}}
\hline
    \multirow{2}{*}{App. Name}
    & \multicolumn{2}{ c | }{Sizes}
    & \multicolumn{4}{ c| }{Performance Statistics}
    & \multicolumn{2}{ c }{Summaries Statistics} \\\cline{2-9}
    &  bytecode instr. & \NumMeth
    & \Processed (\%) & \TotTime (s) & \ClockTime (s) & \TimeMethodAvg (ms)
    & \NumGuarded & \NumUpdates \\\hline

    \href{https://greenrobot.org/eventbus/}{\eventBusname}
    & 2.9k & 177 & 100 & 5.25 & 9.0 & 30 & 10 &   37
    \\\hline

    \href{https://xerial.org/snappy-java/}{\snappyname}
    & 5.6k & 193 & 100 & 10.06 & 14.5 & 52 & 13 &   30
    \\\hline

    \href{https://github.com/aws/aws-secretsmanager-caching-java}{\secretManagername}
    & 35.4k & 3970 & 100 & 129.6 & 457 & 33 & 0 &   188
    \\\hline

    \href{https://www.keycloak.org/}{\keycloackapp}
    & 43.2k & 4441 & 100 & 123.474 & 307 & 28 & 22 &  441
    \\\hline

    \href{https://github.com/zxing/zxing/tree/master}{\zxingapp}
    & 54.4k & 1764 & 100 & 137.4 & 236 & 78 & 1 &  245
    \\\hline

    \href{https://directory.apache.org/kerby/}{\kerbyname}
    & 62k & 4637 & 100 & 162.6 & 596 & 35 & 148 &  711
    \\\hline

    \href{https://bookkeeper.apache.org}{\bookkeepername}
    & 143.7k & 6527 & 99.9 & 204 & 1265 & 26 & 313 &   1006
    \\\hline

    \href{https://rocketmq.apache.org}{\rocketname}
    & 221.5k & 13694 & 99.99 & 1504 & 5396 & 110 & 144 &  1376
    \\\hline

    \bottomrule
\end{tabular}

		\vspace*{-1ex}
	\end{table}



	Our experiments show that \Symmaries scales effectively to handle {large real-life applications}, as shown in Table~\ref{tab//web.app.experiments}.
	These results were obtained using our least-precise heap domain.
	The time and complexity of analyzing an application depend on several factors, including the type of analysis (taint checking or implicit flow analysis), the number of methods, the complexity of the control flow graphs of methods, the number of state variables in the generated SCFGs etc.
    In six applications, all methods were successfully processed (the  column \Processed), while over 99\% of methods were processed in \rocketapp and \bookkeepername. The total number of analyzed methods in our experiments was $\sim$35,000, and \Symmaries successfully computed a summary for \textbf{99.99}\% of them.

	{\TotTime in Table~\ref{tab//web.app.experiments} shows the total effective time taken by \Symmaries for the analysis of an application,  \ClockTime is the associated wall-clock time, and \TimeMethodAvg shows the average time for analyzing a method.
		Although it took longer to construct summaries for some methods, the majority of methods could be completed in less than 2ms. The time for the analysis of a method ranged from less than 1ms to about 9s, with an average of 26ms.
		The effective time for the analysis of applications ranges from 5 seconds to 25 minutes, depending on the size of the application and the number of analyzed methods. However, the wall clock time, which is less accurate, spans from 9 seconds to approximately 90 minutes.
		Our results prove that \Symmaries can scale up to analyze real-life large applications.
	}

	Table~\ref{tab//web.app.experiments} also presents results on the inferred summaries.
    Column~\NumGuarded shows the number of methods whose guard specification is not a tautology, i.e., methods that may cause information leakage depending on the calling context.
    Furthermore, \Symmaries identified several methods with false guards. These cases were mainly observed in methods with complex control-flow structures, where over-approximation led to  restrictive outcomes.
    In addition, column~\NumUpdates reports the number of methods that include an information flow between their arguments and return value.


	\vspace{-0.1in}
	\subsection*{RQ3: Usability}
	Although we did not conduct experimental user studies, several aspects of \Symmaries demonstrate its usability in practice. First, the tool integrates with \soot, a widely used analysis framework for Java bytecode, and requires only minimal configuration from the user: the application or library to analyse (as a \texttt{.jar} file), the list of sinks and sources (as an \texttt{.xml} file), and optional security assumptions for excluded methods (\textsf{.secstubs} files) when method definitions are missing (e.g., in native methods) or portions of the program are too difficult to analyse. Thanks to its summary-based analysis and the concept of symbolic sources, \Symmaries does not require extensive code annotations or the identification of entry points, which are often required for static code analysis. This approach has enabled us to analyze real-world applications with minimal manual effort.

    Second, the tool output is expressed in a concise, human-readable summary language. This design allows users to interpret and manually validate the results in their own context, and facilitates informed decision-making when reusing third-party components or assessing the security of their own code. examples of the output are provided in the paper, for instance in \figurename~\ref{fig::running.example.code}.

	Third, \Symmaries supports a hybrid workflow that balances automation with user control. While the tool is designed to automatically infer security summaries for methods, users may override them through \textsf{.secstubs} files to encode more precise specification. Thanks to this flexibility, even if some large or complex methods cannot be processed, users can exclude them and provide alternative specifications (e.g., treating them as third-party library methods or applying a simpler, flow-insensitive analysis left for future work), while still analyzing the remainder of the application.

	Finally, our experiments on real-world Java API libraries and applications provide indirect evidence of usability. \Symmaries successfully handles projects with up to hundreds of thousands of lines of code within reasonable time and memory limits. The ability to process such large-scale systems indicates that the tool is not only theoretically sound but also practically applicable. Together, these features demonstrate that \Symmaries is usable in practice.


	\subsection*{RQ3: Comparison with State-of-the-Art}
	Several tools have been proposed for specification inference from code, including approaches that infer aliasing specifications for libraries~\cite{eberhardt2019unsupervised,bastani2018active} or taint specifications for methods~\cite{tileria2024docflow,clapp2014mining,icse/StaicuTSMP20}, similar to \Symmaries (See Section~\ref{sec::related.works}). However, many of these tools either target languages other than Java or produce specifications of a different kind, making direct comparison with our summaries challenging. Even among the few tools that support Java, specification formats vary. For example, \cite{bastani2018active} expresses results as path specifications, \cite{eberhardt2019unsupervised} uses dedicated predicates, and \cite{livshits2009merlin} synthesizes sources, sinks, and sanitizers. In contrast, \Symmaries infers security guards as logical expressions, together with aliasing specifications and information flows expressed as sets of assignments, which cannot be directly mapped to these formats.
    The closest related approach is StubDroid~\cite{arzt2016stubdroid}, which extracts data flows between sources and sinks in Android applications. StubDroid automatically generates summaries of Android framework methods for taint analysis to improve the performance of FlowDroid~\cite{Artz2014FlowDroid}. However, it is unsound, does not support implicit flows, and is limited to explicit data-flow extraction. In contrast, \Symmaries captures both implicit flows and aliasing information.
    Therefore, a meaningful comparison with existing approaches would require them to produce comparable specifications to ours for Java methods, which none currently do. Therefore, detailed comparison of output against prior work is not possible.

	\subsection*{Threats to the Validity of Results}
    A threat to the validity of our precision results is that our experiments evaluated only the precision of guards.
    A more comprehensive assessment should also incorporate aliasing relationships in addition to the information flows to provide a fuller picture of precision.
    Moreover, we assumed that the security stubs provided by the user are correct. Since incorrect specifications can directly lead to incorrect tool outputs, this assumption is a common threat in formal analysis tools that depend on user-provided annotations or inputs.

	We ran our scalability experiments with the less precise heap model \LPrecise, which offers better performance compared to more precise heap domains, theoretically at the cost of a lower precision. However, our initial observations in applying the tool on real-life applications show that more precise heap models produce more restrictive guards for methods with intricate control flow structures, resulting in a higher rate of false positives.
	A proper investigation of this claim requires significant efforts to (manually) check the guards of thousands of methods, which we leave as future work.
	Hence, we believe \LPrecise might offer better scalability and precision in practice overall, and can be a reasonable choice to run \Symmaries with.


\vspace{-0.1in}
\section{Related Works}\label{sec::related.works}
\paragraph{Summary-based Analysis}
Summary-based approaches for program analysis offer a solution to mainly address scalability issues, where the goal is to perform modular analysis, incremental analysis~\cite{DBLP:journals/tse/HeYC22}, or even loop summarization.
A summary can capture various types of information depending on the purpose of the analysis, e.g., points-to relations~\cite{sas/NystromKH04}, dataflows~\cite{ccs/WeiLOCZ18}, locking behavior~\cite{usenix/BaiL022}.
UBITect~\cite{sigsoft/ZhaiHZWSQLKY20} uses a summary-based approach to detect UBI bugs in the Linux kernel, where method summaries are twofold like ours: the requirements on input under which the method call is safe and the updates to the outputs with respect to the inputs.
Bai et al. propose a summary-based dataflow analysis to detect locks in the Linux kernel~\cite{usenix/BaiL022}.
The authors in~\cite{kbse/LeeLR20} employ a summary-based approach to detect bugs in multilingual programs where semantic summaries from programs written in the guest language are extracted and used for the analysis of the host language. The focus of these approaches is on detecting specific bugs rather than inferring security specification, and they are not sound.

The few existing summary-based approaches for security analysis focus on taint analysis. Wei et al.~\cite{ccs/WeiLOCZ18} propose a summary-based taint analysis approach to identify the inter-language dataflow in Android applications containing native code. Stievenart et al.~\cite{scam/StievenartR20} propose a summary-based taint analysis for WebAssembly programs, where a function summary illustrates the information flows from parameters and global variables to the return value and the global variables after execution. Zhang et al.~\cite{ccs/ZhangCHLZZQ21} employ Dr. Checker~\cite{uss/MachirySCSKV17}, a framework for detecting bugs in Linux drivers, to construct taint flows for entry functions to detect higher-order taints. While the above approaches focus on taint detection, Staicu et al.~\cite{icse/StaicuTSMP20} focus on a similar goal to ours and propose a summary-based approach to extract taint specifications for JavaScript libraries. This work uses dynamic taint analysis and testing to construct summaries, and their summaries are different from ours in the sense that they do not capture how the procedure affects the heap or when this method can be used. 
Merlin~\cite{livshits2009merlin} extracts information flow specifications using probabilistic inference and inter-procedural data flow analysis, identifying potential sources, sinks, and sanitizers within methods. However, it does not generate specifications that describe method behavior.
In contrast to the above approaches, our method is sound, and our summaries are more expressive: they infer guards as logical expressions, along with information-flow and aliasing specifications represented as assignments. Existing approaches may achieve better scalability, but at the cost of sacrificing soundness.

The summary-based approaches have been used to address scalability of formal verification,
\eg  in software
model-checking ~\cite{Sery:2011:IFS:2425986.2426005},
the
verification of numerical properties in programs via \emph{abstract
	interpretation}~\cite{Cousot:1977:AIU:512950.512973},
and construction of 
procedure summaries via relational abstract interpretation~\cite{
	BoutonnetHalbwachs2019DisjRelAbstrInterp4Summaries}.
  These approaches construct
summaries that are \emph{relations} that link the \emph{values} of
inputs and outputs.
\emph{In contrast, the summaries that we compute are contracts, that we argue are easier to interpret in an IFC analysis context.}
We also provide a guard to show when to use this method securely, in addition to providing the effects.
The shape analysis of ~\citet{Gulavani:2009:BSA:1615441.1615457} (which
discovers structural properties about manipulated data-structures)
also makes use of logic-based summaries, yet they rely on some
fragments and/or extensions of Separation
Logic~\citep{Reynolds:2002:SLL:645683.664578}.
\emph{To the best of our knowledge, there has been no research on using summary-based formal analysis to address security concerns, specifically in the information flow analysis context.}

\paragraph{Library Specification Inference}
Several approaches have been proposed for inferring API specifications. USpec~\cite{eberhardt2019unsupervised} uses unsupervised learning to infer aliasing relations in APIs, which are then used to enhance points-to analysis. Atlas~\cite{bastani2018active} automatically infers points-to specifications using dynamic execution. DAInfer~\cite{wang2024dainfer} derives aliasing relations from API documentation using natural language processing and optimization techniques. Spectre~\cite{kan2025spectre} generates aliasing specifications for library APIs using fuzzing techniques.
The approach in~\cite{eberhardt2019unsupervised} learns specifications by analyzing API usage patterns without requiring manual annotations, access to source code, or the ability to execute the APIs. Ramanathan et al.~\cite{ramanathan2007static} infer preconditions of API calls using predicate mining and inter-procedural data flow analysis.
While some of these methods are related to our work in that they infer aliasing relations, they do not extract security summaries and lack soundness, unlike our approach. Clapp et al.~\cite{clapp2014mining} mine explicit information flows of libraries from concrete executions, and Tileria et al.~\cite{tileria2024docflow} extract taint specifications from API documentation. However, these approaches are unsound, fail to capture implicit flows, and do not consider aliasing in the extracted specifications, unlike our work.

\nbrem[yep]{Few works have suggested using summaries in the context of C(-like) program analysis.
Several approaches ~\citep{Choi:1993:EFI:158511.158639,
Yu:2010:LLM:1772954.1772985} construct summaries as (partial) transfer functions that summarise the side effects of each procedure, 
whereas Saturn~\citep{Xie:2005:SED:1040305.1040334
} achieves relatively scalable and precise constraint-based analyses by using custom graphs to construct summaries.
\textsc{Calysto}~\citep{Babic:2008:CSP:1368088.1368118} also builds summaries for performing static checking (\aka linting) of C programs.}


\paragraph{IFC Verification}
Static analysis approaches for ensuring noninterference have been extensively studied in the research community. The vast majority of proposed information flow control solutions focus on type systems~\cite{pottier2003information,sabelfeld2003language,Barthe:2007:CLN:1762174.1762189}, and various tools targeting realistic programming languages have been developed to verify such properties—e.g., JFlow/JIF~\cite{jif-myers-popl99} and \KeY~\cite{KeY}. Other approaches rely on more general static analysis techniques, such as interprocedural dataflow analysis~\cite{Reps:1995:PID:199448.199462} and program slicing~\cite{Kam:1977:MDF:2696874.2696926}. Notable frameworks in this category include JOANA~\cite{Hammer:2009:FCO:1667545.1667547}, DroidSafe~\cite{GordonKPGNR15DroidSafe}, and \textsc{FlowDroid}~\cite{Artz2014FlowDroid}.
There are also practical IFC solutions tailored for web applications~\cite{Moller:2014:ADC:2668018.2531921,hedin2014jsflow} and Android apps~\cite{li2017static}. However, most of these either do not handle implicit flows or lack soundness. When soundness is enforced, it often comes at the cost of practicality—for example, by rejecting any program that contains a method call in a high-security context~\cite{Lortz2014Cassandra}, or by requiring extensive manual annotations, as in \cite{KeY,jif-myers-popl99}.
\emph{While the goal of these approaches is vulnerability detection through static analysis, our focus lies in inferring abstract representations of security-relevant behavior.} Among these, only JFlow/JIF~\cite{jif-myers-popl99} infers security labels for variables. However, our approach introduces more expressive guards that can include conditions on heap structures too. Further, in contrast to \Symmaries, JIF does not compute method updates and relies heavily on manual annotations.

\green{
\paragraph{Open-Source Software Security}
Static security analysis and vulnerability detection in open-source code aim to determine whether a program adheres to specific security requirements. They categorize code as either secure or insecure based on a provided security policy, which might differ from the consumer's policies. However, these methods do not analyse the security behaviour of code for arbitrary security policies or express its security implications if reused.
{Software reuse strategies~\cite{cacm/RavichandranR03} can be categorized into either a white-box approach, where developers have access to the component content and can modify it, or a black-box approach, wherein the code is utilized 'as is,' and developers does not access to the code.
Software development relies on package managers and code repositories, enabling developers to share code and reuse third-party code.}
Approaches to identify vulnerabilities in open-source code can be categorized into two types: version-control-based approaches that check whether a specific version of a code clone contains reported vulnerabilities~\cite{ccs/DuanBXKL17,icse/WooPKLO21}, or code-based methods where malicious or vulnerable code is detected by being matched to known vulnerable code~\cite{uss/Xiao0YXYLL0HZS20,sp/KimWLO17}.
V1SCAN~\cite{uss/WooCLO23} is a hybrid method that combines version-based and code-based approaches to detect vulnerabilities in C/C++ programs using classification techniques.
There are several methods to identify modified code; for example, ~\cite{icse/WooPKLO21} proposes the approach \textsc{Centris} to identify modified OSS reuse, and ~\cite{icse/ReidJM22} developed the tool VDiOS to find vulnerabilities that have already been patched in the original projects.
Further, existing works to ensure supply chain integrity either focus on detecting and isolating malicious packages~\cite{ZimmermannSTP19,StaicuPL18}, or on designing trusted package managers~\cite{TorresAriasAKC19,KuppusamyDC17}.
Software composition analysis (SCA) is a process to identify the open source code in a  program.
 \emph{Our approach is complementary to the above approaches, in the sense that
 these approaches focus on detecting how vulnerable code has been propagated, while we take a different approach to provide abstractions of the program security behaviour to help the developer  in choosing a secure open-source code to use.}
{%
Further, existing approaches for secure open-source code development often leverage package managers, enabling developers to share and reuse third-party code securely. Existing efforts primarily concentrate on either detecting and isolating malicious packages~\cite{ZimmermannSTP19,StaicuPL18} or designing trusted package managers to ensure supply chain integrity~\cite{TorresAriasAKC19,KuppusamyDC17}. However, these approaches often lack formal foundations necessary to provide sound security guarantees, and they do not explain the security behavior of code in contrast to ours.
}
}

\begin{leaveout}
Static security analysis and vulnerability detection in open-source code aim to determine whether a program adheres to specific security requirements. They categorize code as either secure or insecure based on a provided security policy, which might differ from the consumer's policies. However, these methods do not analyse the security behaviour of code for arbitrary security policies or express its security implications if reused. This presents a challenge for developers in assessing its suitability for use.
	Therefore, for the secure reuse of (open-source) code, it is essential to adopt approaches with solid foundations to identify the security implications of code correctly and ensure compliance with security requirements. These approaches should describe the security-related behaviour of the code in an abstract way, enabling developers to make informed decisions about the code's suitability for secure reuse and understanding its security implications if utilized.

\todo{See the survey on supply chain attacks.}


\todo{Add FlowDist}
\end{leaveout}

\section{Conclusions}
\label{sec::conclusions}
We have presented a scalable approach for constructing sound security summaries of Java bytecode methods. Each summary consists of a guard, specifying the conditions under which a method call is secure, and an update, capturing the resulting changes to aliasing and information flow. Our tool-supported technique leverages co-reachability analysis to compute these summaries. An evaluation of \Symmaries demonstrated its effectiveness, precision, and scalability, showing that it can successfully analyze large, real-world applications.

    As ongoing work, we are predicting methods that are difficult to analyze and excluding them from automatic analysis while providing a manual summary to improve the precision and scalability of our approach (i.e., to save time by avoiding analysis of methods that are less likely to complete successfully).
	\Symmaries currently employs a naive name-based resolution procedure
	to construct an over-approximation of the call-graph and solve
	Eq.~(\ref{eq:full-fixpoint}).
	In that regard, several avenues exist to improve the precision of our
	tool for real-life applications by employing alternative, sound yet
	more precise, call-graph construction
	algorithms~\cite{Reif2019JudgeCGConstruction}.
	Furthermore, bottom-up program analyzes are generally
	well-suited to parallel implementations, and we plan to improve the
	run-time performance of \Symmaries in this way.
	Another work consists in extending our approach to generate runtime monitors to control
	information flow.
	This consists in augmenting our input language with checkpoint
	statements, for which we synthesize security conditions that are
	checked at runtime~to~ensure~security.

%
\ifacmartloaded
\bibliographystyle{ACM-Reference-Format}
\bibliography{bibs/ifc,bibs/sa,bibs/mc,bibs/sr,bibs/des,bibs/rv,bibs/other}
\fi
\ifacmartloaded
\appendix
\fi
\green{

    \section{Soundness Proof of Summary Computation}\label{sec::summary.computation.proof}

	Let \(\Psi_0 ~\Psi_1 \ldots \Psi_n \) be a sequence of execution of the procedure $\Summary{.}{.}$ to compute the IF-summaries of methods $\Methods$  where \(\Psi_i = (m_i, \MethSums_i)\), \({m_i \in \Methods}\) is the method for which the IF-summaries are computed in the step $i$ and \(\MethSums_i\) is the set of all IF-summaries computed so far for all the methods $m_i$, \({0 \leq i \leq n}\).
	To prove monotonicity of $\Summary{.}{.}$ \wrt
	\( \langle\Summaries m, ⊑_{\Summaries m} \rangle \), we should prove that \(\MethSums_i(m)  ⊑_{\Summaries m} \MethSums_j ({m})\) holds for any method $m$ where ${0 \leq i<j \leq n}$.
	To prove this, we induct on $n$. The base case is obvious. Let this property hold for any pairs up to the index $k$, \ie  $\{0<i,j \leq k\}$. To prove the property for ${n=k+1}$, two cases can happen.
	If there exists no \(\Psi_i= (m_{k+1}, \MethSums_i), i <k+1\), then the property obviously holds.
	Otherwise, let $m_{k+1}$ be processed in a step $i < k+1$.
	According to the induction step,  \(\MethSums_i ({m'}) ⊑_{\Summaries m} \MethSums_j ({m'}) \) holds for any method $m'$ called by $m$ where $0 \leq i<j \leq k$. Let  \(\MethSums_{j} ({m})  \not\sqsubseteq_{\Summaries m} \MethSums_{k+1} (m)\), \ie the synthesized guard {\SumGuard {m,k+1}} for $m$ in the step $k+1$ is more permissive than {\SumGuard {m,j}} computed in $j$ and
	\(\sState \not\models \SumGuard {m,j} \) and  \(\sState \models \SumGuard {m,k+1} \) for some state $\sState$.
	The invariant in rule \rname{Call} is more restrictive in the step $k+1$ according to the inductive step hypothesis, \ie  $ \sState \models \SumGuard {m',j} \implies \sState \models \SumGuard {m',k+1}$ for all states \sState.
	If \({\MethSums_j(m') = \MethSums_{k+1}(m')}\) for all $m' \in \Methods $ called by $m$, this will obviously
	imply \({\MethSums_j(m) = \MethSums_{k+1}(m)}\) which contradicts our assumption.
	If there is a callee method $m'$ where \(\MethSums_j(m') \sqsubset_{\Summaries m} \MethSums_{k+1}(m')\), the synthesized guard \(\SumGuard {m,k+1} \) should avoid some states leading to a state in which $m'$ is called, \ie there is a state $\sState$ from which executing $m$ will lead to the new states $\sState'$ of $m$ restricted by the invariant ${\SumGuard {m',k+1}}$ while \(\sState' \models \SumGuard {m',j} \). This means that \(\sState \not\models \SumGuard {m,k+1} \) should hold to avoid the bad state $\sState'$ which contradicts our assumption and proves  \(\MethSums_j ({m}) ⊑_{\Summaries m} \MethSums_{k+1} ({m})\).

	\section{Intra-procedural Security Semantics}\label{sec::intra.sematics}
\figurename~\ref{fig:security-semantics-for-statements} shows the simplified security semantics of assignment and branch statements. $f_p$ and $f_r$ represent a primitive field and a reference field, respectively. The function $\sf updHeap (s)$ updates the heap structure and the security level of objects affected by the execution of $s$.
The variable $\omega$ abstracts away the actual branch condition. The function \TBrch{\region{}\stm} is used for upgrade analyses to handle implicit flows. For a detailed description of the security semantics of intra-procedural statements, see \cite{BerthierKhakpour23Heap}.

\begin{figure}[!h]
  \centering
  \begin{freeruleset}
    \myfreerule{\GotoRule}{%
    }{%
      \semloc{\s{\cgoto~l};\stms}
      \trans{\tt, ∅}%
      \newloc{\Target{} l}%
    }
    \qquad
  \end{freeruleset}
  \\[\myrulespace]
  \begin{ruleset}
    \myrule{\AssignRule}{%
      \begin{array}{@{}l@{}}
	T_\stm =
	\left\{%
	\begin{array}{@{}l@{\,}c@{\,}l@{\ \text{if~}\stm = \,}l@{}}
	  ［\plvl v \assignv \plvl e］ &&& \s{v = e}\\
	  ［\plvl v \assignv \plvl r ⊔ \flvl[\hvar]r{f_p}］ &&& \s{v = r.f_p} \\
	  ［\plvl r \assignv \plvl s ⊔ \flvl[\hvar]s{f_r}］ &\SQMergei u& \LoadRef r s {f_r}{\hvar}& \s{r = s.f_r} \\
	  ［\plvl r \assignv \plvl s］ &\SQMergei u& \CopyRef r s {\hvar} & \s{r = s}\\
	  ［\plvl r \assignv ⊥］& \SQMergei u& \NullRef r {\hvar} & \s{r = \Null}\\
	  ［\plvl r \assignv ⊥］& \SQMergei u& \NewRef r c {\pc}{\hvar} & \s{r = \cnew~c}\\
	  \multicolumn 3 {r@{\ \text{if~}\stm = \,}} {%
	  \StorePrim  r {f_p} e {{\nomlvl{\plvl{e}}}} {\hvar}} & \s{r.f_p = e} \\
	  \multicolumn 3 {r@{\ \text{if~}\stm =  \,}} {%
	  \StoreRef r {f_r} s {\nomlvl{\plvl s ⊔ \hlvl[\hvar] s }}{\hvar}} & \s{r.f_r = s}
	\end{array}\right.
      \end{array}%
    }{%
      \semloc{\stm;\stms}%
      \trans{\tt, T_\stm}%
      \newloc{\stms}%
    }\\[\myrulespace]%
    \myrule{\BranchRule}{%
      \quad
      ℓ = \s{\cif~(e)~\cgoto~l};\stms
    }{%
      \begin{array}{@{}l@{~}l@{}@{~}l@{}}%
        ℓ \trans{\phantom{¬}ω ∧ ¬\uVar ∧ \plvl e \,\not⊑\, \pc, \TBrch{\region{}\stm}} \newloc{\Target{} l} \quad\quad &
	ℓ\trans{\phantom{¬}ω ∧ ¬\uVar ∧ \plvl e \,⊑\, \pc ∨ \phantom{¬}ω ∧ \uVar, ∅} \newloc{\Target{} l} \\
        ℓ\trans{        ¬ ω ∧ ¬\uVar ∧ \plvl e \,\not⊑\, \pc, \TBrch{\region{}\stm}} \newloc{\stms} \quad\quad  &
        ℓ\trans{        ¬ ω ∧ ¬\uVar ∧ \plvl e \,⊑\, \pc ∨           ¬ ω ∧ \uVar, ∅} \newloc{\stms}
      \end{array}%
    }
  \end{ruleset}
  \par
  \raggedright%
  \vspace{\myrulespace}
  %
  \caption{Security semantics of intra-procedural statements.}
  \label{fig:security-semantics-for-statements}
\end{figure}

\section{Full Semantics}\label{sec::full.sematics}
\begin{figure*}
  \begin{minipage}{\linewidth}
  \begin{ruleset}
    \myrule{\mStmRule}{%
      \begin{array}{@{}c@{}}
        \semloc{\stm;\stms\;} %
        \trans{\phi, T}%
        \semloc{\;\stms'}
         \quad
        \stm ∈ \{  \begin{array}{@{}c@{}}
        	\s{\creturn}, \s{\creturn\,x},  
        	\s{\coutput l (\_)}
        \end{array} 
    		\}
        \quad
        (\pval', \mHeapVal') = \evalStm(\stm,\pval, \mHeapVal)
      \end{array}
    }{%
      \mframe[]{\stm;\stms}\pval \mHeapVal:: \mStack%
      \ttrans{\neg \mode ∧ \phi} T%
      \mframe[]{\stms'}{\pval'}{ \mHeapVal'}:: \mStack%
    }\\[\myrulespace]
    \myrule{\mCallRule}{
      \begin{array}{@{}c@{}}
        \stm = \s{[\_ =]\,r.m'(\lits)} \quad
         ℓ_m = \semloc{\stm;\stms}
     \quad
        {ℓ_m \trans{\tt,T}%
        \stms'} \\
        \methFrame = \mframe[]{ℓ_m}{\pval} \mHeapVal\quad
        \methFrame[m'] = {\NewFrame[\stm, \pval, \mHeapVal]}
      \end{array}
    }{
      \begin{array}{@{}c@{}}%
        \methFrame :: \mStack
        \ttrans{\neg \mode 
        }{\emptyset}
        \methFrame[m'] :: \methFrame :: \mStack
      \end{array}
    }\\[\myrulespace]
    \myrule{\mExitRule}{%
      \begin{array}{@{}c@{}}
      	 \stm = \s{[\_ =]\,r.m'(\lits)} \quad
        ℓ_m = \semloc{\stm;\stms} \quad  ( \pval_m',\mHeapVal_m') = \evalStms{(\mathsf{footUpd},\pval_m, \mHeapVal_m)} \quad
        {\stm };\stms \trans{\tt, T}\newloc{\stms}
      \end{array}
    }{%
      \begin{array}{@{}c@{}}%
        \mframe[]{\surd}{\pval_{m'}}{\mHeapVal_{m'}} :: \mframe[]{ℓ_m}{\pval_m}{\mHeapVal_m} :: \mStack
        \ttrans{
        \tt}{T}
        \mframe[]{\stms}{\pval_m'}{\mHeapVal_m'} :: \mStack%
      \end{array}
    }
  \end{ruleset}
  \par
  \smallskip%
  \raggedright%
  where
  \smaller
  \[
  \mathsf{footUpd}= \{r = \this \}\cup \{{w^{-1}} (p) = p~|~ p \in {\mathit{Args}_m \cap \MethRefs m}\}\cup      \begin{cases}
  	{x=\pi} & 	\text{if~} \stm = \s{x = r.m(w)}\\
   \emptyset & \text{if~} \stm = \s{r.m(w)}
  \end{cases}
  \]
  \begin{align*}
     \evalStm{(a, \pval, \mHeapVal)} ≝ \,
    \begin{cases}
      (\pval[π⟼\pval(v)], \mHeapVal) &\mbox{if~} a = \{\creturn~v\}\\
      (\pval, { \mHeapVal[π⟼ r]} )&\mbox{if~} a=\{\creturn~r\}\\
      (\pval, \mHeapVal )& \mbox{if~} a \in 
      \{\coutput\_(\_),\creturn\} \\
    \end{cases}
  \end{align*}
  and\vspace*{-.2\baselineskip}%
  \begin{align*}
	\evalStms{(A, \pval, \mHeapVal)} ≝ \,
	\begin{cases}
		(\pval, \mHeapVal) &\mbox{if~} A = \emptyset\\
		\evalStms{(A \backslash \{a\}, \pval', \mHeapVal')} &a \in A \mbox { and }	(\pval', \mHeapVal') = \evalStm{(a, \pval, \mHeapVal)}
	\end{cases}
\end{align*}
\end{minipage}
\caption{Full semantics of inter-procedural statements.}
  \label{fig::full.semantics}
\end{figure*}

%

The execution semantics of intra-procedural statements is defined according to the semantics introduced in~\cite{BerthierKhakpour23Heap}.
We limit ourselves to the definition of the inter-procedural-related  statements in this paper.
In the full semantics, a \emph{configuration} is defined as a non-empty stack  \( \methFrame[m_0] :: \methFrame[m_1]:: \ldots :: η\) of \emph{frames} where each frame represents the local state for a method invocation.
We define a \emph{frame} for a non-terminated call to a method \(m\) as a tuple \(\methFrame[m] ≝ \mframe[]{ℓ_m}{\pval_m}{\mHeapVal_m}\) where $ℓ_m$ denotes the semantic location in its basic semantics,  \(\pval_m
\) is a {\emph{valuation}} for all primitive variables used in $m$'s body, 
and $\mHeapVal_m$ is the concrete heap value.
We give the full semantics in \figurename~\ref{fig::full.semantics}, where \(\trans{\phi,T}\) and \(\ttrans \phi T\) are transitions of the basic semantics and of the full semantics, respectively.
Note that the updates by the basic semantics are captured in  transition updates while the updates by the execution semantics are explicitly shown in target locations.
The transitions of the full semantics
use
the variables of the SCFG associated with the frame at the top of the stack in the source configuration.
We refer to an SCFG that specifies the full semantics of a method $m$ by  \(\G_{m}\) built using the above semantics rules.
The rule \mStmRule handles 
return and sink statement, by extending the updates by the basic semantics with updates to the method's memory \pval and its concrete heap \mHeapVal (we use a dedicated return variable π to hold returned values until the sink location is reached).
When a method is called, a new context to execute the callee $m'$ is created with the help of the \NewFrame function, and pushed to the stack (\rname{m-CALL}); this may only happen when the method at the top of the stack runs in nominal mode (\ie \(\mode = \ff\)).
This frame basically consists of a tuple \(\methFrame[m'] = \mframe[]{\stms_{m'}} {\pval_{m'}} {\mHeapVal_{m'}}\) where \MethBody{m'} is the body of the target method \(m'\) which is determined based on the class of the object reference by \(r\), \ie according to \emph{dynamic dispatch}, \(\pval_{m'}\) is a valuation for \(m'\)'s local variables and arguments that is appropriately initialized with the effective arguments \lits, and $\mHeapVal_{m'}$ is the initialized concrete heap.
Upon creation of \methFrame[m'], an SCFG \(\G_{m'}\) is instantiated 
 with initial location \semloc{\stms_{m'}}.
Its state variables are initialized according to the initialization equation of basic semantics,
with the addition of a constraint that
	\begin{enumerate*}[(i)]
		\item assigns every symbolic state variable from \MethFormalArgs{m'} with the current valuation of corresponding state variables in \(\G_m\),
		\item initializes the symbolic heap abstraction for \(m'\) \wrt the effective arguments \lits, and
		\item copies the security level of the calling context \(\pc\) to that of the callee.
	\end{enumerate*}
	The semantic location for frame \methFrame only changes when the sink location of the callee is reached
	(rule \rname{m-Exit}), where the top frame is popped from the stack and the control returns back to the caller method by updating
	the base reference \(r\),  the reference arguments and the returned value, if any according to  the function $\mathsf{methUpd}$.
Observe that $w^{-1}$ is the inverse of the  mapping function $w$ which maps that variables in the callee to the caller variables.
{It also upgrades the typing environment of the caller according to the callee's summary update (See the rule \rname{Call}).}

 \renewcommand{\Symmaries}{Anonymized}

\section{Proof of Noninterference}
\label{sec::noninterference-proof}

\begin{definition}[Compatible States]\label{def::compatible.states}
	We say two states \( \qstate_1\) and \( \qstate_2\) are {compatible},
	denoted by \(\qstate_1 \approx \qstate_2 \), iff
	\begin{enumerate}[(1)]
		\item\label{compatible.states.item:same.locations} \(\qstate_1(ℓ)=\qstate_2(ℓ)\),		  		
		\item\label{compatible.states.item:same.memory.typing.environment} \(\qstate_1({\VarsTypes})=\qstate_2(\VarsTypes)\),
		\item\label{compatible.states.item:same.heap.typing.environment} \(\qstate_1(\HeapTypes)=\qstate_2(\HeapTypes)\),
		\item\label{compatible.states.item:low.equiv.memory} \(\qstate_1(\pval)=_{\low}\qstate_2(\pval)\),
		\item\label{compatible.states.item:low.quiv.heap}  \({\qstate_1(\mHeapVal) =_{\low} \qstate_2(\mHeapVal)}\), and
			\item\label{compatible.states.item:related.stacks}  \( \qstate_1( \mStack) =  \qstate_2( \mStack) \)
	\end{enumerate}
\end{definition}

\begin{proof}
  To prove Theorem~1, we need to find a witnessing bisimulation relation $\InRelation$ that witnesses $\qstate_1 \lowbisim \qstate_2$. We define $\InRelation$ as the following and prove that it is a bisimulation relation:
  \[
    \InRelation=\left\{ \langle \qstate,\qstate' \rangle~|~
      \begin{array}{@{}ll@{}}
        \qstate \approx \qstate' \wedge \beta(\qstate,\qstate')
      \end{array}%
    \right\}
  \]
  where \begin{equation}		\label{compatible.states.item:same.pc.mode}
    \beta(\qstate,\qstate') =
    \begin{array}{l}
    \qstate(\mode)=\qstate'(\mode)=\ff ~\wedge\\
\qstate(\pc) =\qstate'(\pc) = \low .
    \end{array}
 \end{equation}

 { The initial states $\qstate_0$ and $\qstate'_0$ are obviously in the relation $\InRelation$ according to Definition~1 (Low-Bisimulation), {$\qstate_0\models{ x_{0\mbox-\mathit{all}}}$ and $\qstate'_0\models{ x_{0\mbox-\mathit{all}}}$}.}
  Let $\langle \qstate, \qstate' \rangle \in \InRelation$. According to the definition of low-bisimulation, we should prove
  %
  that,  if $\qstate\xrightarrow{o}_* t$,
  then either (a) there exists $t'$ such that $\qstate' \xrightarrow{o}_* t'$
  and $t \InRelation t'$,
  or (b) $\qstate' \xrightarrow{\low}_* $,
  and vice versa.
  If no observation is made from $\qstate'$, then the case (b) obviously holds.
  If an observation is made after $\qstate'$, then the conclusion is followed from Lemma~\ref{lemm::single.observation} which states that $t$ and  $t'$ are in the relation \(\InRelation\) too.
\end{proof}


\begin{lemma}[Single Observation]\label{lemm::single.observation}
  Let   \(\qstate_1 \InRelation \qstate_2 \).
  If $\qstate_1	\xrightarrow{o}_* \qstate'_1$
  and	$\qstate_2	\xrightarrow{o}_* \qstate'_2 $, then \(\qstate'_1 \InRelation \qstate'_2 \).
\end{lemma}

\begin{proof}
  We prove this by induction on the length of $\qstate_1	\xrightarrow{o}_* \qstate'_1$.
  Since, the execution does not terminate in $\qstate_1$, this means \(\qstate_1(\ell) \neq \surd \) and \(\qstate_1(\ell) = \stm;\stms  \).\\
  \textbf{Base Case} The base case (\ie the length of one) happens when $\stm$ is an output statement like \coutput l (\elts).
\item	In case of \coutput l (\elts), \plvl \elts can not be high-sensitive according to $φ_ℓ$  in the rule \OutputRule
  Hence, according to the conditions \ref{compatible.states.item:low.equiv.memory} and \ref{compatible.states.item:same.memory.typing.environment} in Definition~\ref{def::compatible.states}, the observations should be the same in both states.
  Since it does not update the state, apart from the symbolic location, the target states will remain compatible as well.

  \noindent
  \textbf{Inductive Step}  Let it hold for all executions $\qstate_1	\xrightarrow{o}_* \qstate'_1$ with the length of \(m\) or less than \(m\). We should prove the lemma for an execution with the length of  \(m+1\).
We perform a case analysis on \stm . The proof of intra-procedural statements is same as the corresponding ones  in \cite{BerthierKhakpour23Heap}. We limit ourselves to the proof of inter-procedural statements.
  \begin{enumerate}[(i)]

  \item Let \stm  be a method call \(r.m'(\lits)\).  The invariant \rname{φ-Call}  in the rule \textsc{Call}
  holds, otherwise the states would have been avoided by the synthesized guards. Therefore, according to the rule \textsc{M-Call}
  , a new call context is created and initialized.
    The initialization comprises executing a sequence of non-interfering assignments.
    We can prove that the two states after applying each assignment individually will still remain in relation according to Lemma~\ref{lemm:assignments}.
    {Since the assignments by the initialization step are non-interfering, it's trivial to show that the whole updates by all the assignments will preserve the relation. The conclusion is derived by applying the inductive step hypothesis on the final states of the frame initialization step. }

  \item
    Let \stm  be an output statement  \coutput l (\elts). Since these states have not been avoided by the security guards,  this means that the invariant  \rname{φ-Sink}  in the rule \OutputRule
    holds in $\qstate_i, i \in \{1,2\}$.
    However, observations should be produced in the last transition of this execution according to the definition of $	\xrightarrow{o}_* $. Since, this transition is not the last one, this means that \stm cannot produce any observation, \ie the invariant
     \rname{φ-Sink}
     does not hold and \stm cannot be an output statement.

  \item Let \stm be \creturn~[x]. 	Let $\qstate_1	\to t_1$  and $\qstate_2	\to t_2$.
    {The return variable \sVar is updated  to $x$ in both states,} \ie the new states will still remain in relation.

  \item Let \stm be an exit point, \ie
  \[\qstate_i = \mframe[]{\surd}{\pval_{m}}{\mHeapVal_{m}} :: \mframe[]{[a=]r.m'(w)}{\pval_{m'}}{\mHeapVal_{m'}} :: \mStack\]
  for $i \in \{1,2\}$.
    From \(\qstate_1 \InRelation \qstate_2 \) and the condition~\ref{compatible.states.item:related.stacks} in Definition~\ref{def::compatible.states}, it follows that the stack will be the same in both states, \ie  \( \qstate_1(\mStack) = \qstate_2(\mStack) \). Hence, from this and the fact that the top frame will be popped from the stack without any further changes to it, we can conclude that \(\qstate'_1(\mStack) = \qstate'_2(\mStack)\) too (I).

    According to the rule \mExitRule, \this, the reference parameters $w^{-1} (p)$, and the possible assignee $a$ will be updated.
    The typing environment of the calling method is upgraded according to the callee's effect too.
    This means a sequence of non-interfering assignments will be executed.
    Similar to the proof of Lemma~\ref{lemm:assignments}, we can prove that Eq.~\ref{compatible.states.item:same.pc.mode} and the conditions \ref{compatible.states.item:same.locations}-\ref{compatible.states.item:low.quiv.heap} in Definition~\ref{def::compatible.states} hold for the pair \((\qstate'_1 , \qstate'_2)\) (II).

    Therefore, from (I) and (II), we can conclude that $\qstate'_1 \InRelation \qstate'_2$.
  \end{enumerate}
\end{proof}

\begin{lemma}[Assignment Preserves Relation]\label{lemm:assignments}
  Let   \(\qstate_1 \InRelation \qstate_2 \) where \(\qstate_i(ℓ) = \stm;\stms, i \in \{1,2\} \) and \stm is an assignment. If $\qstate_1	\to \qstate'_1$ and $\qstate_2	\to \qstate'_2 $, then \(\qstate'_1 \InRelation \qstate'_2 \).
\end{lemma}

\begin{proof}
See \cite{BerthierKhakpour23Heap}.
\end{proof}

\newcommand{\context}{\ensuremath{\mathsf{context}}}

\newcommand{\regions}{\ensuremath{\Lambda}}

}

\end{document}